\pdfoutput=1
\documentclass[a4paper,12pt]{article}

\usepackage{graphicx}
\usepackage{amsmath,amssymb}
\usepackage{fancybox,ascmac}
\usepackage{amsthm}
\usepackage{mathrsfs}
\usepackage{braket}
\usepackage{mathtools}
\usepackage{tikz}
\usetikzlibrary{intersections, calc, arrows.meta}
\usepackage{physics}
\usepackage{hyperref}

\makeatletter
\@addtoreset{equation}{section}
\makeatother

\makeatletter 

\@addtoreset{figure}{section}
\makeatother 

\makeatletter 

\@addtoreset{table}{section}
\makeatother 

\DeclareMathOperator{\diag}{diag}
\DeclareMathOperator{\id}{id}
\DeclareMathOperator{\im}{Im}
\DeclareMathOperator{\lcm}{LCM}

\newtheorem{thm}{Theorem}[section]
\newtheorem{dfn}{Definition}[section]
\newtheorem{prop}{Proposition}[section]

\newtheorem{asmp}{Assumption}[section]
\newtheorem{lem}{Lemma}[section]

\title{On the Rationality and the Code Structure of a Narain CFT, and the Simple Current Orbifold}
\author{Yuma Furuta\thanks{Reserch Institute for Mathematical Sciences, Kyoto Univercity, 
Kyoto, Japan}
\thanks{Email: yfuruta@kurims.kyoto-u.ac.jp}}
\date{\today}

\begin{document}
\maketitle

\begin{abstract}

    In this paper, we discuss the simple current orbifold of a rational Narain CFT 
    (Narain RCFT). This is a method of constructing other rational CFTs from a given 
    rational CFT, by ``orbifolding'' the global symmetry formed by a particular primary 
    fields (called the simple current). Our main result is that a Narain RCFT satisfying 
    certain conditions can be described in the form of a simple current orbifold of 
    another Narain RCFT, and we have shown how the discrete torsion in taking that 
    orbifold is obtained. 
    Additionally, 
    the partition function can be considered a simple current orbifold with discrete 
    torsion, which is determined by the lattice and the B-field. We establish that the 
    partition function can be expressed as a polynomial, with the variables substituted 
    by certain $q$-series. In a specific scenario, this polynomial corresponds to the 
    weight enumerator polynomial of an error-correcting code. 
    Using this correspondence to the code theory, we can relate the B-field, the 
    discrete torsion, and the B-form to each other.

\end{abstract}

\newpage
\tableofcontents

\newpage
\section{Introduction}

The holographic correspondence between the three-dimensional 
TQFT and the two-dimensional RCFT has produced some very fascinating studies, 
both mathematically and physically. In particular, the TQFT behaves like a topological 
sector of three-dimensional gravity and can be a convenient touchstone in the 
construction of a quantum gravity theory. Such correspondence has been investigated 
mathematically, and a representative example is the well-known construction of RCFT 
on the boundary for TQFTs with 3-dimensional defects, as discussed in \cite{fuchs2002,fuchs2004}. 
Such a story translates on the 2D side as a simple current orbifold. This orbifold 
is a method proposed by Kreuzer and Schellekens to construct another RCFT from the RCFT, 
and it constructs a modular invariant partition function by specifying the gluing 
of the irreducible representations of the right and left chiral algebra. However, 
as we will see in section \ref{sec2.3}, 
this orbifold is not uniquely defined, but has a degree of freedom called 
a discrete torsion. The reason why it is called orbifold is that it is an 
operation that ``orbifolds'' the global symmetry formed by a particular primary fields 
(called simple current).

In this paper, we discuss the physical meaning of the discrete torsion and 
its counterpart in code theory, as described below. In doing so, we focus on CFTs 
created from Lorentzian even self-dual lattices, i.e., Narain CFTs. First, regarding 
the simple current, as pointed out in \cite{frohlich2007}, it is known that any RCFT can be obtained as 
a simple current orbifold of a diagonal modular invariant RCFT. Here, diagonal modular 
invariant is a partition function of RCFTs consisting of the same representations of 
the right and left chiral algebra. Since such an RCFT is known to be the case of Narain 
CFT with zero antisymmetric B-field, it is expected that the nontrivial simple current 
orbifold discrete torsion has some relation to the nonzero B-field. However, 
no clear assertion of such a relationship is known yet. 
An explicit relation between these two objects can provide a physical interpretation 
of the simple current orbifold, and it is a very interesting question how this can be 
expressed in terms of a 3-dimensional defect line. We have shown that the Narain RCFT 
can be regarded as a simple current orbifold of a diagonal modular invariant when 
certain conditions are satisfied, and we have succeeded in expressing the discrete 
torsion in terms of the coordinates ($\Lambda, B$) 
of the Narain moduli in concrete terms.

In conjunction with these results, we also found an analogue of this discrete 
torsion to the error-correcting code, which is simply the B-form of the code. In short, 
it is the B-form of the code, and as will be discussed in detail in Section \ref{sec2.2}, such a 
correspondence provides a code-theoretic translation of the simple current orbifold. 
An interesting aspect of this finding is its connection with our previous paper\cite{furuta2022}. 
We previously found that the CFT spectral gap is in a sense equivalent to the EPC distance when 
translated into code theory. This distance is a known quantity in combinatorics and 
can be calculated unambiguously, and the necessary conditions for the spectral gap 
to be large can be expressed in a different expression in a straightforward 
manner. In fact, since this EPC distance can be calculated from the B-form of a 
binary code, it is expected that this object retains some information on the CFT 
spectrum. Therefore, if we generalize this correspondence and compare it 
with our present results, we may be able to extract information on the spectral 
gap of the Narain RCFT from the discrete torsion. In this sense, our results are 
very interesting.

In particular, we focus on the relationship between the weight enumerator polynomial 
and the partition function. In other words, we would like to ask how the partition 
function can be reconstructed for a given Narain CFT.
The motivation behind investigating such phenomena lies in the holographic 
correspondence between 3-dimensional topological theories and 2-dimensional RCFTs. 
Recently, spurred by the holographic correspondence between 2-dimensional JT gravity 
theories and 1-dimensional CFTs, the relationship between d-dimensional gravity theories 
and theories obtained by taking ensemble averages of $(d-1)$-dimensional CFTs has been 
explored. In particular, 3-dimensional quantum gravity theories are believed to be 
topological theories, and connections between ensemble averages of 3-dimensional 
Chern-Simons theories and 2-dimensional Narain CFTs were discovered by 
Maloney and Witten, among others \cite{maloney2020,afkhami-jeddi2021,ashwinkumar2021}.


In particular, although the averaging for rational theories satisfying certain 
properties has been addressed in various papers \cite{meruliya2021,aharony2023}, the summation for all rational 
Narain CFTs has not been computed. Therefore, this study aims to express the 
partition function of rational Narain CFTs as a simple polynomial of combinatorial 
objects, execute the averaging for all rational theories by replacing it with the 
combinatorial summation of polynomials, and determine their holographic behavior. 
In this paper, we demonstrate that the partition function of rational Narain CFTs 
satisfying specific conditions can be expressed as a ``weight enumerator polynomial," which 
is an extension of error-correcting codes. By unveiling the presence of ``discrete 
mathematical" objects behind rational theories, we aim to achieve a clearer 
understanding of holographic correspondences.

\section{Narain CFT and Simple Current Orbifold}
In this section we briefly review some of the concepts necessary to
describe our results.
\subsection{Definition of a Narain CFT}\label{sec2.1}
First, we introduce the Narain CFT. This is the CFT that arises when $n$ 
bosons are compactified into an $n$-dimensional torus, but there is also an 
antisymmetric B-field in the background.
The action of this CFT is given by
\begin{equation}
    \mathcal{S}=\frac{1}{4\pi\kappa}\iint\dd{\sigma}\dd{\tau}\left(
        h^{\alpha\beta}\partial_{\alpha}X^{I}\partial_{\beta}X_{I}
    +2B_{IJ}\dot{X}^I X'^J
    \right),
\end{equation}
where the bosons $X(\tau,\sigma)$ are compactified as
\begin{equation}
    X \sim X+2\pi e ,\ e\in \Gamma. 
\end{equation}
Here $\Gamma$ is the lattice that specifies the torus and has generator matrix $\gamma$
, and $h^{\alpha\beta}=\mqty(\dmat{1,-1}),\ I=0,\dots,n-1$.
Furthermore, the momentum vectors $(\vec{p}_L,\vec{p}_R)$ can be written 
explicitly using B and $\gamma$ as
\begin{equation}
    \begin{split}
        \vec{p}_L&=\frac{2\vec{P}+(B+I)\vec{e}}{2}\\
        \vec{p}_R&=\frac{2\vec{P}+(B-I)\vec{e}}{2},
    \end{split}
\end{equation}
for $\vec{P}\in \Gamma^{\ast}$. 
Therefore, $(\vec{p}_L,\vec{p}_R)^{\top}$ form an even self-dual lattice with Lorentzian
metric 
$g=\begin{pmatrix*}
    I & 0 \\ 0 & -I
\end{pmatrix*}$ that has the generator matrix of form 
\begin{equation}\label{eq2.4}
    \Lambda = \begin{pmatrix}
        \gamma^{\ast} & \frac{I+B}{2}\gamma \\
        \gamma^{\ast} & \frac{-I+B}{2}\gamma 
    \end{pmatrix}.
\end{equation}

Since all the vertex primary fields are of the form 
$
    V_{\vec{p}_L,\vec{p}_R}=:e^{i\vec{p}_L\cdot X_L+i\vec{p}_R\cdot X_R}:
$
, the partition function of a Narain CFT is 
\begin{equation}\label{eq2.5}
    \mathcal{Z}(\tau ,\bar{\tau})=\frac{1}{|\eta|^{2n}}\sum_{(\vec{p}_L,\vec{p}_R)\in\Lambda}
    q^{\frac{|\vec{p}_L|^2}{2}}\bar{q}^{\frac{|\vec{p}_R|^2}{2}}
    ,\ q=e^{2\pi i\tau},\ \bar{q}=e^{2\pi i\bar{\tau}}.
\end{equation}
with $\eta=q^{\frac{1}{24}}\prod_{i=1}^{\infty}(1-q^i)$.

Remark that there is a kind of dualities that connects physically equivalent
theories, named T-duality.
This duality consists of left $O(n)\times O(n)$ actions and right $O(n,n;\mathbb{Z})$
actions, i.e.
\begin{eqnarray}
    \Lambda \sim O_1 \Lambda ,\ O_1 \in O(n)\times O(n)\\
    \Lambda \sim \Lambda O_2 ,\ O_2 \in O(n,n;\mathbb{Z}).
\end{eqnarray}
Therefore the moduli space of Narain CFTs with $c=n$ is 
\begin{equation}
    \mathcal{M}_n = O(n)\times O(n)\backslash O(n,n)/O(n,n;\mathbb{Z}).
\end{equation}

\subsubsection{Rationality of a Narain CFT}\label{sec2.1.1}
We would like to comment here on the rationality of the Narain CFT, 
which is very important to our story.
Put simply, a CFT is rational if the Hilbert space of the theory is equal to the direct 
sum of a finite number of irreducible characters of the current algebra.
In particular, many code CFTs are rational. As mentioned above, this is a mathematically 
and physically well-behaved class of theories, and the necessary and sufficient 
conditions for Narain CFTs to be rational have already been investigated by Wendland\cite{wendland2000moduli}.

\begin{thm}[A part of Theorem 4.5.2 of \cite{wendland2000moduli}]\label{thm2.1}
Let $\mathcal{C}(\Lambda,B)$ denote a Narain CFT with central charge $c=d$.
Then $\mathcal{C}$ is rational if and only if $G\coloneqq\Lambda^{\top}\Lambda\in\text{GL}(d,\mathbb{Q})$
and $B\in\text{Skew}(d)\cap\text{Mat}(d,\mathbb{Q})$.
\end{thm}
Here by $\text{Skew}(d)$ we mean the set of $d\times d$ antisymmetric matrixes. 
Thanks to this theorem, determining whether the Narain CFT is rational or 
not has become straightforward. Therefore, we will proceed to ascertain the 
conditions of the theorem\ref{thm2.1} under which Narain CFTs, specifically 
Narain RCFTs, can be constructed through error-correcting codes. In this pursuit, 
it is imperative for us to explicitly derive the possible combinations of primary 
field irreducible characters that constitute the partition function.
To achieve this goal, we employ the technique of simple current orbifolds as 
elucidated in \ref{sec2.3}.

\subsection{Brief Review of Construction A}\label{sec2.2}
This subsection primarily describes how to construct Narain CFT using codes.
This construction was initially discovered by Dymarsky and Shapere \cite{dymarsky2020a}, and various 
extensions have been considered.
What all of these extensions have in common is that there is a correspondence 
between the properties of the codes and the properties of the lattice.
Thus, Narain CFT can be analyzed using codes, and it can also be transformed 
into discrete mathematical objects \cite{furuta2022}.
In this article, We will begin by the more general $\mathbb{F}_p$ code as a code and show 
how stabilizer codes appear as a special case of the code.
Here we will focus on the construction A, which enables us to 
relate the weight enumerator polynomial to the partition function explicitly.

Let $\mathcal{C}$ be a code over $\mathbb{F}_p$ with length $2n$ and dimension $n$,
i.e., $\mathcal{C}$ has the $n\times 2n$ generator matrix G with each element in 
$\mathbb{F}_p$.
In addition, define the dual code $\mathcal{C}^{\perp}$ as 
\begin{equation}
    \mathcal{C}^{\perp}=\{c'\in \mathbb{F}_p \mid c\cdot c'=0 \text{in}\ \mathbb{F}_p \},
\end{equation}
and we say $\mathcal{C}$ is self-dual when $\mathcal{C}=\mathcal{C}^{\perp}$.
Evenness of $\mathcal{C}$ is defined to be the case that all the elements of C are even.
Note that here we employed the Lorentzian metric 
$\begin{pmatrix}
    0 & I \\ I & 0
\end{pmatrix}$
where $I$ is $n\times n$ identity matrix.
For a codeword $c$, we denote the first $n$ elements of $c$ by $\alpha$
and the other by $\beta$, so that $c=(\alpha, \beta)$.
We are now ready to define the lattice constructed from $\mathcal{C}$.
\begin{dfn}
    Let $\mathcal{C}$ be a code over $\mathbb{F}_p$ with length $2n$.
    Define the lattice $\Lambda_{\mathcal{C}}$ by
\begin{equation}
    \Lambda_{\mathcal{C}}=\left\{ \left(\frac{\alpha +pk_1}{\sqrt{p}},\frac{\beta+pk_2}{\sqrt{p}}\right) 
    \mid c=(\alpha,\beta)\in\mathcal{C},\ (k_1,k_2)\in\mathbb{Z}^{2n}\right\}.
\end{equation}
\end{dfn}
Then one finds that if $\mathcal{C}$ is even self-dual, so is $\Lambda_{\mathcal{C}}$ with respect
to the Lorentzian metric 
$\begin{pmatrix}
    0&I\\I&0
\end{pmatrix}$.
Therefore, one can construct a Narain CFT from a even self-dual code.

In particular, one can bring the generator matrix of the $\mathbb{F}_p$-code 
to the standard form as 
\begin{equation}\label{eq2.10}
    G=\begin{pmatrix}
        I & B^{\top}
    \end{pmatrix}
\end{equation}
where $B$ is $\mathbb{F}_p$-valued $n\times n$ antisymmetric matrix.
If $\mathcal{C}$ has a generator matrix of form eq.(\ref{eq2.10}), we say
$\mathcal{C}$ to be B-form.
For a B-form $\mathcal{C}$, the generator matrix of $\Lambda_{\mathcal{C}}$
is following:
\begin{equation}\label{eq2.11}
    \Lambda=\begin{pmatrix}
        pI & B \\ 0 & I 
    \end{pmatrix}/\sqrt{p}.
\end{equation}
We abused the notation of the generator matrix of a lattice to be the same
as the lattice itself. 
This Narain CFT corresponds to the $n$ bosons compactified on cubic lattice
with size $\sqrt{\frac{2}{p}}$ and B-field $B$.
It should be noted that the metric of $\Lambda_{\mathcal{C}}$ differs from 
the one of eq.(\ref{eq2.4}).
To bring these forms into alignment, rotate the vectors as 
\begin{equation}
    \left(\frac{\alpha +pk_1}{\sqrt{p}},\frac{\beta+pk_2}{\sqrt{p}}\right)^{\top}
    \mapsto \frac{1}{\sqrt{2}}\begin{pmatrix}
        I & I \\ I & -I
    \end{pmatrix}
    \left(\frac{\alpha +pk_1}{\sqrt{p}},\frac{\beta+pk_2}{\sqrt{p}}\right)^{\top},
\end{equation}
and they will coincide with $(\vec{p}_L,\vec{p}_R)^{\top}$.
Explicitly, they are related as 
\begin{equation}\label{eq2.13}
    \begin{pmatrix}
        \vec{p}_L \\ \vec{p}_R
    \end{pmatrix}
    =\frac{1}{\sqrt{2}}
    \begin{pmatrix}
        \frac{\alpha +\beta +p(k_1+k_2)}{\sqrt{p}}\\
        \frac{\alpha -\beta +p(k_1-k_2)}{\sqrt{p}}
    \end{pmatrix}.
\end{equation}

It should be noted that there are $O(1,1)$ degrees of freedom 
in this codeword to momentum construction as noted in \cite{angelinos2022}. 
In other words, 
the Narain lattice metric is preserved even if an element of 
$O(1,1)$ is applied to each codeword-derived vector component by component. 
Therefore, this redundancy should be taken into account in the discussion, 
but it is not considered in this paper because it can be discussed without 
considering this redundancy.

To close this subsection, we will mention that for $p=2$, 
the code can be regarded as a quantum stabilizer code.
First, define the Hilbert space $\mathcal{H}$ of qubit states as $(\mathcal{C}^{2})^{\otimes n}$, 
and its basis $\{\ket{x_1\dots x_n}\mid x_i \in \mathbb{F}_2 \}$.
Then a $k$ dimensional subspace of $\mathcal{H}$ which we denote by $\mathcal{H}_{\mathcal{C}}$
is called stabilized subspace, if a group of Pauli operator $\mathscr{S}$
stabilize all elements of $\mathcal{H}_{\mathcal{C}}$. We now 
suppose that each element of $\mathscr{S}$ can be written in terms of 
Pauli operators $X$ and $Z$, where
\begin{equation}
    X :\ket{x}\mapsto\ket{x+1},\ Z : \ket{x}\mapsto (-1)^{x}\ket{x}.
\end{equation}
Then we can write each generator $s_i$ of $\mathscr{S}$ explicitly as 
\begin{equation}
    s_i=i^{\alpha_i \cdot\beta_i}\varepsilon\left(
        X^{(\alpha_i )_1}\otimes\cdots\otimes X^{(\alpha_i)_n}
    \right)\left(
        Z^{(\beta_i)_1}\otimes\cdots\otimes Z^{(\beta_i) _n}
    \right),\ \alpha_i,\beta_i \in \{0,1\} ^n
\end{equation}
where $\varepsilon$ is a constant phase factor. 
As can be easily seen, $\mathscr{S}$ must be commutative for $\mathscr{S}$ to 
stabilize the code subspace.
This condition is replaced by the following equation using $\alpha$ and $\beta$:
\begin{equation}\label{eq2.16}
    ^{\forall}i\neq ^{\forall}j,\ (\alpha_i,\beta_i)
    \begin{pmatrix}
        0 & I \\ -I & 0
    \end{pmatrix}
    \begin{pmatrix}
        \alpha_j ^{\top}\\
        \beta_j ^{\top}
    \end{pmatrix}=0 \mod 2.
\end{equation}
Introducing the generator matrix of $\mathcal{C}$ as 
\begin{equation}
    G=\begin{pmatrix}
    \alpha_1 & \beta_1 \\
    \vdots & \vdots \\
    \alpha_{n-k} & \beta_{n-k}   
    \end{pmatrix},
\end{equation}
the condition \eqref{eq2.16} becomes
\begin{equation}
    GWG^{\top}=0 \mod 2,\ W=
    \begin{pmatrix}
        0 & I \\ -I & 0
    \end{pmatrix}.
\end{equation}
If a classical code $\mathcal{C}$ is even self-dual, 
the generator matrix $G_{\mathcal{C}}$ of $\mathcal{C}$
satisfies 
\begin{equation}
    G_{\mathcal{C}}g G_{\mathcal{C}}^{\top}=0\ \mod 2, g=
    \begin{pmatrix}
        0 & I \\ I & 0
    \end{pmatrix}.
\end{equation}
Note that $g$ and $W$ are equivalent modulo 2.
Therefore, when C is a classical even self-dual code, it readily 
forms a stabilizer code as well.
Of course, since $g$ and $W$ are not equal in general mod $p$, not all 
$\mathbb{F}_p$ codes can be regarded as stabilizer codes. However, by considering a 
special class of codes, it is possible to discover codes that are both self-dual 
and commutative \cite{kawabata2022}.

\subsection{Simple Current Orbifold}\label{sec2.3}
In this subsection, we describe the simple current orbifold developed by 
Kreuzer and Schellekens in \cite{kreuzer1994}. This method of constructing another RCFT from one 
RCFT is governed by a degree of freedom known as discrete torsion. In other words, 
the simple current orbifold is constructed using a decomposition of the partition 
function in terms of characters, where the discrete torsion determines how the 
characters are combined.
Here we assume that the RCFT is a Narain CFT.

First, we introduce the simple current.
A chiral primary operator $\mathcal{O}_{\vec{p}_L}$ is referred to as 
simple current, if for every fusion product with a primary field 
it produces single primary field.  
In fact, all primary fields of the chiral part of a Narain CFT are simple current. 
This is because the fusion product of any two operators is given by the addition of 
lattice vectors,
\begin{equation}
    \mathcal{O}_{\vec{p_1}_L}(z)\times \mathcal{O}_{\vec{p_2}_L}(0)
    \sim z^{-\vec{p_1}_L \cdot\vec{p_2}_L}
        \mathcal{O}_{\vec{p_1}_L +\vec{p_2}_L}(z).
\end{equation}
In particular, the fusion product is commutative.
We consider the case that a set of simple currents $\{J_1,\dots ,J_k\}$ constitute
a commutative algebra $K=\mathbb{Z}_{N_1}\times\cdots\times\mathbb{Z}_{N_k}$.
Then, given a Verlinde algebra $\mathcal{V}$, one can write down the modular invariant partition function as 
\begin{equation}
    \mathcal{Z}(\tau,\bar{\tau})=\sum_{a\in\mathcal{V},\beta\in K}M_{a,\beta}
    \chi _a \bar{\chi}_{[J]^{\beta}a}
\end{equation}
where $\chi$ are characters of the chiral algebra and 
$[J]^{\beta}a$ means $J_1^{\beta _1}\cdots J_k^{\beta _k}a$. 
Here the coefficients $M_{a,\beta}$ are determined by the following equation
\footnote{Originally, the constant $\text{Mult}(a)$ is multiplied to the right hand side, 
but since it is always $1$ in this paper, it is omitted.}
\begin{equation}\label{eq.2.23}
    M_{a,\beta}=\prod_{i}\delta ^1(Q_i (a)+X_{ij}\beta _j),
\end{equation}
and we will describe what are $Q_i(a)$ and $X_{ij}$.
$Q_i (a)$ is called monodromy charge, defined using the conformal dimension $h(\mathcal{O})$ of 
a primary field $\mathcal{O}$ by 
\footnote{Although this definition is different by the sign from \cite{kreuzer1994},
we have inverted the sign to conform to the notation in \cite{fuchs2004}.}
$Q_i (a)=h(J_i a)-h(J_i )-h(a)$. 
$X$ is defined by the sum of two $n\times n$ matrices. One is the symmetric matrix $R_{ij}$, 
defined by
\begin{equation}\label{eq2.24}
    h(J_1 ^{\alpha_1}\cdots J_k ^{\alpha_k})=
    \frac{1}{2}\sum_{i,j} \alpha_i R_{ij}\alpha_j +
    \frac{1}{2}\sum_{i} r_{ii}\alpha_i\ \mod 1
\end{equation}
where $R_{ij}=\frac{r_{ij}}{N_i}\ (r_{ij}\in\mathbb{Z})$. The other is the antisymmetric matrix $e_{ij}$, 
which plays the role of 
degrees of freedom when the orbifold is performed.
Specifically, it is expressed as $e_{ij}=\frac{\epsilon_{ij}}{N_i},\ (\epsilon_{ij}\in\mathbb{Z})$.
Then the sum $\frac{R}{2}+e$ is written as $X$.
Thus, the equation \eqref{eq.2.23} implies that any combination of $\alpha$ and $\beta$ 
such that the equation
\begin{equation}\label{eq2.25}
    Q_i(a)+X_{ij}\beta_j =Q_i(a)+\left(\frac{R_{ij}}{2}+e_{ij}\right)\beta_j = 0\ \mod 1\ 
    \text{for}\ \forall i
\end{equation} 
is satisfied will appear in the character decomposition of the partition function.
Therefore, we call equation \eqref{eq2.25} the twist equation, and in the following section 
we will investigate a 
certain simple current orbifold of the Narain CFT by studying the solution of the 
twist equation.

Before going to the next section, we will give an easy example, which 
is the code CFT constructed from the $p=2$ B-form code.  
First, consider the B-form code with the generator matrix 
\begin{equation}
    G=\begin{pmatrix}
        1 & 0 & 0 & 0 \\
        0 & 1 & 0 & 0 
    \end{pmatrix}.
\end{equation}
Then the partition function of the code CFT constructed from the above code 
can be written as 
\begin{equation}
    \mathcal{Z}_0 (\tau ,\bar{\tau})=\sum_{i,j=0}^{3}\chi_{(i,j)}\bar{\chi}_{(-i,-j)},\ 
    \chi_{(i,j)}=\chi_i \chi_j ,\ \bar{\chi}_{(i,j)}=\bar{\chi}_i \bar{\chi}_j
\end{equation}
where $\chi_i$ and $\bar{\chi}_{-i}$ are 
\begin{equation}
    \chi_i (\tau)=\frac{1}{\eta (\tau)^2}\sum_{n\in\mathbb{Z}}q^{2(n+\frac{i}{4})^2},\ 
    \bar{\chi}_i (\bar{\tau})=\frac{1}{\bar{\eta (\tau)}^2}\sum_{m\in\mathbb{Z}}\bar{q}^{2(m+\frac{i}{4})^2}
\end{equation}
with $q=e^{2\pi i\tau}$, $\bar{q}=e^{2\pi i\bar{\tau}}$.
Therefore one finds the Verlinde algebra is $K=\mathbb{Z}_4 \times\mathbb{Z}_4$
and can consider simple currents generated by $\mathcal{O}_{(2,0)}$ and 
$\mathcal{O}_{(0,2)}$ (Here the subscripts represent the labels of the chiral 
algebra.). These chiral primary operators consists of the $\mathbb{Z}_2 \times \mathbb{Z}_2$
subgroup of $K$, so one can take the simple current orbifold of this theory 
with these operators.
In this case, the monodromy charge $Q_i(a)$ of a operator $\mathcal{O}_{(a_1 ,a_2)}\ (a=(a_1,a_2)\in K)$
becomes the inner product of two momentum vector $\vec{p}_1 \cdot\vec{p}_2$
since the conformal dimension is given by $|\vec{p}|^2 /2$.
In this case, the matrix $R_{ij}$ can easily verified to be identity matrix $\delta_{ij}$
, and we set the discrete torsion $e=\begin{pmatrix}
    0 & \frac{1}{2} \\ -\frac{1}{2} & 0
\end{pmatrix}$
, which is the only nontrivial one with $n=2$.
Therefore one obtains the twist equation
\begin{equation}\label{eq2.29}
    \frac{2\cdot a_i}{4}+\frac{\delta_{ij}+\epsilon _{ij}}{2}\beta_{j}=0 \mod 1
\end{equation}
for $i=1,2$.
Since $\beta$ takes the value in $\mathbb{Z}_2 \times\mathbb{Z}_2$, 
the solution of \eqref{eq2.29} is found to be 
$(a_1,a_2)=(0,0),(0,2),(2,0),(2,2)$ for $(\beta_1 ,\beta_2)=(0,0),(1,1)$,
and $(a_1,a_2)=(1,1),(1,3),(3,1),(3,3)$ for $(\beta_1 ,\beta_2)=(1,0),(0,1)$.
As a remark, if we fix $\beta$ then there always exist some solutions of the 
twist equation, however, there does not always for fixed $a$.
Then the partition function after taking the simple current orbifold 
can be written as 
\begin{equation}
    \mathcal{Z}_{\text{orb}}(\tau,\bar{\tau})=
    |\chi_{(0,0)}'|^2 +|\chi_{(1,1)}'|^2 + \chi_{(1,0)}' \bar{\chi}_{(0,1)}'
    +\chi_{(0,1)}' \bar{\chi}_{(1,0)}'
\end{equation}
where $\chi_{(i,j)}'(i,j\in\mathbb{Z}_2)$ are 
\begin{equation}
    \begin{split}
        \chi_{(0,0)}'&= \chi_{(0,0)}+\chi_{(2,2)},\ \chi_{(1,1)}'=\chi_{(2,0)}+\chi_{(0,2)},\\
        \chi_{(1,0)}'&= \chi_{(1,1)}+\chi_{(3,3)},\ \chi_{(0,1)}'=\chi_{(1,3)}+\chi_{(3,1)}.
    \end{split}
\end{equation}
This partition function is same as the one of the code CFT constructed from B-form 
code 
\begin{equation}
    G' = \begin{pmatrix}
        1 & 0 & 0 & 1 \\ 0 & 1 & 1 & 0
    \end{pmatrix},
\end{equation}
which is the tensor square of a boson compactified on a circle with radius $R=\sqrt{2}$.
With this example, one would expect that the discrete torsion $e$ is 
related to the B-form of the underlying B-form code. In the following section,
we will explain how these notions are connected, by explicitly ``reconstructing'' the 
underlying code.

\section{Main Theorems}
This section is dedicated to describing our main findings, a portion of which elucidates 
when a Narain CFT is rational and exhibits a code structure. Naturally, this condition 
aligns with Wendland's work \cite{wendland2000moduli}, but we also propose that by imposing additional 
conditions, a rational Narain CFT can be viewed as a simple current orbifold model of 
another CFT, where the discrete torsion can be represented by the compactifying lattice and
the antisymmetric B-field. We prove this statement at the theorem \ref{thm3.1}, 
and then consequently we claim that the B-form of an error-correcting code 
is related to the B-form,
which means that essentially these three objects (the B-field, 
the discrete torsion, and the B-form)  
can be seen as equivalent objects to each other.

\subsection{Decomposition of the Partition Function}\label{sec3.1}
First, we will prove that if the B-form and compactifying lattice satisfy 
certain conditions then the Narain CFT is rational. This will be done by 
decomposing the partition function into a finite sum 
of characters of primary fields.

In the following, we will consider the Narain CFT after 
performing a T-duality on the Narain lattice \eqref{eq2.4}, which acts
$\begin{pmatrix}
    0 & I \\ I & 0
\end{pmatrix}\in O(n,n;\mathbb{Z})$
from right and 
$\begin{pmatrix}
    I & 0 \\ 0 & -I
\end{pmatrix}\in O(n)\times O(n)$ from left
where $I$ is the $n\times n$ identity matrix.
Then we obtain the Narain lattice 
\begin{equation}
    \Lambda=
    \begin{pmatrix}
        \frac{I+B}{2}\gamma & \gamma^{\ast} \\
        \frac{I-B}{2}\gamma & -\gamma^{\ast}
    \end{pmatrix},
\end{equation}
so the momentum vectors $(\vec{p}_L ,\vec{p}_R)^{\top}$ can be expressed 
using integer vectors $(N,M)^{\top}\in (\mathbb{Z}^n)^2$
as 
\begin{equation}\label{eq3.2}
    \begin{pmatrix}
        \vec{p}_L \\ \vec{p}_R
    \end{pmatrix}
    =
    \begin{pmatrix}
        \frac{I+B}{2}\gamma & \gamma^{\ast} \\
        \frac{I-B}{2}\gamma & -\gamma^{\ast}
    \end{pmatrix}
    \begin{pmatrix}
        N \\ M
    \end{pmatrix}
    =
    \begin{pmatrix}
        \frac{I+B}{2}\gamma N + \gamma^{\ast}M\\
        \frac{I-B}{2}\gamma N - \gamma^{\ast}M
    \end{pmatrix}.
\end{equation}
As we have obtained the momentum vectors, the partition function 
\eqref{eq2.5} becomes
\begin{equation}\label{eq3.3}
    \mathcal{Z}(\tau,\bar{\tau})=\frac{1}{|\eta(\tau)|^{2n}}
    \sum_{N,M\in\mathbb{Z}^n}
    q^{\|\frac{I+B}{2}\gamma N + \gamma^{\ast}M\|^2 /2}
    \bar{q}^{\|\frac{I-B}{2}\gamma N - \gamma^{\ast}M\|^2 /2}.
\end{equation}
Noting that $\gamma^{\ast}=(\gamma^{\top})^{-1}$, 
the half of the vector norm $\frac{\|\vec{p}_L \|^2}{2}$ can be deformed into 
\begin{equation}\label{eq3.4}
    \frac{\|\vec{p}_L \|^2}{2}
    =\left( \frac{\gamma ^{\top}\gamma +\gamma^{\top} B\gamma}{2}N+M \right)^{\top}
    (\gamma^{\top}\gamma)^{-1}
    \left( \frac{\gamma ^{\top}\gamma +\gamma^{\top} B\gamma}{2}N+M \right).
\end{equation} 
Then, if the matrix $\gamma^{\top}\gamma$ is diagonal, the vector norm \eqref{eq3.4}
can be decomposed into a sum of each component of 
$\frac{\gamma ^{\top}\gamma +\gamma^{\top} B\gamma}{2}N+M$
weighted by the diagonal elements of $\gamma^{\top}\gamma$.
Therefore, in this paper we assume that the diagonal elements of $\gamma^{\top}\gamma$
are rational;
\begin{asmp}\label{asmp3.1}
For the compactifying lattice $\gamma$, the matrix $\gamma^{\top}\gamma$
is diagonal and rational in the sense that 
\begin{equation}
    \gamma^{\top}\gamma = \diag\left(\frac{2t_1}{s_1},\dots ,\frac{2t_n}{s_n}\right)
\end{equation}
where $s_i,t_i \in \mathbb{Z}_{>0}$ and 
$\text{GCD}(s_1,t_1)=\cdots =\text{GCD}(s_n ,t_n)=1$.
\end{asmp}
This assumption leads us to the decomposition of the vector norm \eqref{eq3.4}
\begin{equation}\label{eq3.6}
    \frac{\|\vec{p}_L\|^2}{2}=
    \sum_{i=1}^{n}\frac{s_i}{4t_i}
    \left( \frac{\gamma ^{\top}\gamma +\gamma^{\top} B\gamma}{2}N+M \right)_i ^2.
\end{equation}
Here $\left( \frac{\gamma ^{\top}\gamma +\gamma^{\top} B\gamma}{2}N+M \right)_i$
means the $i$-th component of 
$\left( \frac{\gamma ^{\top}\gamma +\gamma^{\top} B\gamma}{2}N+M \right)$.
Computing the character decomposition of the partition function,
we will start with the simple case for $\gamma^{\top}\gamma =\diag\left(
    \frac{2t}{s},\dots ,\frac{2t}{s}
\right)=
\frac{2t}{s}I$, for which we call the homogeneous case.

\subsubsection{Homogeneous Case}\label{sec3.1.1}
To begin with, we will calculate the Verlinde algebra for $B=0$.
Since $\gamma^{\top}\gamma$ is homogeneous, the vector norm \eqref{eq3.6}
is found to be 
\begin{equation}
    \frac{s}{4t}\left\| \frac{t}{d}N + M \right\|^2
    =
    st\left\| \frac{1}{2s}N + \frac{1}{2t}M \right\|^2,
\end{equation} 
and similarly for $\frac{\|\vec{p}_R \|^2}{2}$,
\begin{equation}
    \frac{\|\vec{p}_R \|^2}{2}=st\left\| \frac{1}{2s}N - \frac{1}{2t}M \right\|^2.
\end{equation}
Then,
one can explicitly write the partition function as 
\begin{equation}
    \mathcal{Z}(\tau ,\bar{\tau})=
    \frac{1}{|\eta(\tau)|^2}
    \sum_{N,M}\prod_{i=1}^{n}q^{st\left(\frac{N_i}{2s}+\frac{M_i}{2t}\right)^2}
    \bar{q}^{st\left(\frac{N_i}{2s}-\frac{M_i}{2t}\right)^2}.
\end{equation}
Rewriting $\frac{N_i}{2s}+\frac{M_i}{2t}=m_i + \frac{l_i}{2st}$ for 
$m_i \in\mathbb{Z},\ 0\leq l_i <2st\ (1\leq i\leq n)$, this leads to 
\begin{equation}\label{eq3.10}
    \begin{split}
        &\mathcal{Z}(\tau ,\bar{\tau})\\
        &=
    \frac{1}{|\eta(\tau)|^2}
    \sum_{\vec{l}=(l_1,\dots ,l_n)\in\mathbb{Z}_{2st}^n}
    \sum_{\substack{(m_1,\dots ,m_n)\in\mathbb{Z}^n\\(m'_1,\dots ,m'_n)\in\mathbb{Z}^n}}
    \prod_{i=1}^{n}q^{st\left(m_i+\frac{l_i}{2st}\right)^2}
    \bar{q}^{st\left(m_i ' +\frac{(2t-1)l_i}{2st}\right)^2}\\
    &=
    \sum_{\vec{l}\in \mathbb{Z}_{2st}^n}
    \chi ^{(2st)}_{\vec{l}}\bar{\chi}^{(2st)}_{\vec{l}'},\ 
    l'_i = (2t-1)l_i, 
    \end{split}
\end{equation}
with $\chi^{(2st)}_{\vec{l}}=\chi^{(2st)}_{l_1}\cdots\chi^{(2st)}_{l_n}$ being 
the $n$-fold product of irreducible character of a boson compactified on radius 
$R=\sqrt{\frac{2t}{s}}$ \cite{difrancesco1997a}.
Thus, one can verify that the Verlinde algebra for $B=0$ is $\mathbb{Z}_{2st}^n$.
Then we will show that a specific choice of the nontrivial $B$ 
corresponds to a simple current orbifold whose discrete torsion is 
related to $B$.
For demonstrating a simple current orbifold, we choose a set $S$ of simple current that is 
generated by 
$\{\mathcal{O}_{(2t,0,\dots ,0)}
,\dots ,\mathcal{O}_{(0,\dots ,0,2t)}\}$, which forms $\mathbb{Z}_s ^n$ symmetry. 
These simple currents correspond to $\{J_1,\dots ,J_n\}$ in the section \ref{sec2.3}
where $k=n$.
In order to clarify the statement, we set the proposition we are going to prove:
\begin{prop}\label{prop3.1}
Let $\mathcal{A}$ be a Narain CFT with its compactifying lattice satisfying the 
assumption \ref{asmp3.1} and homogeneous. 
If the following equation 
\begin{equation}\label{eq3.11}
    \left(\frac{\gamma^{\top}B\gamma}{2}\right)_{ij}=\frac{\epsilon _{ij}}{s}
\end{equation}
holds for some $\epsilon _{ij}\in\mathbb{Z}$, 
then $\mathcal{A}$ is identical to the simple current orbifold by $S$ of 
$B=0$ case \eqref{eq3.10}, whose discrete torsion is 
$\frac{\epsilon _{ij}}{s}$.
\end{prop}

As described previously, the former part of this proposition is included in the 
condition of the theorem 4.5.2 of \cite{wendland2000moduli}.
Therefore Narain CFTs satisfying this condition are guaranteed to be 
rational. 
To this theorem we give a new interpretation known as simple current orbifold.
In what follows, we suppose that the equation \eqref{eq3.11} holds.
To prove the proposition \ref{prop3.1}, it is necessary to calculate 
$\frac{\|\vec{p}_L\|^2}{2}$ and $\frac{\|\vec{p}_R\|^2}{2}$.
In this case, these norms are expressed as
\begin{equation}\label{eq3.12}
    \begin{split}
        \frac{\|\vec{p}_L\|^2}{2}&=st\left\| \frac{N}{2s} +\frac{\epsilon N}{2st} 
        +\frac{M}{2t} \right\|^2 \\
        \frac{\|\vec{p}_R\|^2}{2}&=st\left\| \frac{N}{2s} -\frac{\epsilon N}{2st} 
        -\frac{M}{2t} \right\|^2
    \end{split}
\end{equation}
using the matrix $(\epsilon)_{ij}=\epsilon_{ij}$, 
so the proposition \ref{prop3.1} can be proved by 
decomposing the momenta \eqref{eq3.12} into some irreducible characters
and showing that all the combination of labels of left and right Verlinde algebras
follow all the solutions of the twist equation.
The simple current orbifold implies that 
when the label of the left Verlinde algebra is $\vec{l}$, one of the right algebra is
$-\vec{l}-2t\vec{\beta}$.
Therefore if the non-integer parts of the R.H.S. of \eqref{eq3.12} represent
the labels of the left and right Verlinde algebra $\mathbb{Z}_{2st}^n$,
$\vec{\beta}=-N$ because $-\left(\frac{N}{2s} +\frac{\epsilon N}{2st} 
+\frac{M}{2t}\right)-\left(\frac{N}{2s} -\frac{\epsilon N}{2st} 
-\frac{M}{2t}\right)=\frac{2t\cdot (-N)}{2st}$. 
Accordingly, we have two lammas:
\begin{lem}\label{lem3.1}
Let $N,M\in\mathbb{Z}^n$ be fixed, and let $\vec{m}=(m_1,\dots ,m_n)^{\top}$,
$\vec{m}'=(m'_1,\dots ,m'_n)^{\top}\in \mathbb{Z}^n$, and $\vec{l}=(l_1,\dots ,l_n)^{\top}\in
\mathbb{Z}_{2st}^n$ satisfy $\frac{N}{2s} +\frac{\epsilon N}{2st} 
+\frac{M}{2t}=\vec{m}+\frac{\vec{l}}{2st}$, $\frac{N}{2s} -\frac{\epsilon N}{2st} 
-\frac{M}{2t}=\vec{m}'+\frac{-\vec{l}+2tN}{2st}$.
Then the pair $(\vec{l},-N)$ is one of the solution of the equation for $(\vec{r},\vec{\beta})$,
\begin{equation}\label{eq3.13}
    \frac{\vec{r}}{s}+
    \left(\frac{t}{s}I+\frac{\epsilon}{s}\right)\vec{\beta}=0\ \mod 1.
\end{equation}  
In addition, this equation is equivalent to the twist equation of 
the simple current orbifold of $S$ with discrete torsion $\frac{\epsilon}{s}$.
\end{lem}
\begin{proof}
    Since L.H.S. the equation \eqref{eq3.13} is defined up to $\mathbb{Z}$,
    this is equal to  
    \begin{equation}
        2t\cdot \left(\vec{m}+\frac{\vec{r}}{2st}
                \right)
        +\left(\frac{t}{s}I+\frac{\epsilon}{s}\right)\vec{\beta}.
    \end{equation}
Substituting $(\vec{l},-N)$ for $(\vec{r},\vec{\beta})$ and using 
$\frac{N}{2s} +\frac{\epsilon N}{2st} 
+\frac{M}{2t}=\vec{m}+\frac{\vec{l}}{2st}$, one has 
\begin{equation}
    2t\left(\frac{N}{2s} +\frac{\epsilon N}{2st} 
    +\frac{M}{2t}\right)-
    \left(\frac{t}{s}I+\frac{\epsilon}{s}\right)N
    =0\ \mod1,
\end{equation}
which implies the former part of the statement.

Then we will prove that the equation \eqref{eq3.13} is of the same form of 
the twist equation.
First, we will find the monodromy charge $Q_i(\vec{r})$,
which is simply the inner product of two momentum vectors $\vec{p}_L \cdot\vec{p}'_L\ \mod 1$.
Thus $Q_i(\vec{r})$ can be described by the product of two labels of the left Verlinde algebra,
and one can show this as follows.
For simplicity, consider the first component $Q_1(\vec{r})$. This is given 
by the monodromy charge of $\mathcal{O}_{(2t,0,\dots ,0)}$ with the primary operator
of the Verlinde algebra
label $\vec{r}=(r_1 ,\dots ,r_n)^{\top}$.
We take the integer vectors $N,M,N',M',\vec{m},\vec{m}'$ such that 
\begin{equation}
    \begin{split}
        \frac{N}{2s}+\frac{\epsilon N}{2st}+\frac{M}{2t}&=
        \vec{m}+\frac{(2t,0,\dots ,0)^{\top}}{2st}\\
        \frac{N'}{2s}+\frac{\epsilon N'}{2st}+\frac{M'}{2t}&=
        \vec{m}'+\frac{\vec{r}}{2st},
    \end{split}
\end{equation}
and $N,M,N',M'$ label the left momentum vector $\vec{p}_L,\ \vec{p}'_{L}$ as 
eq. \eqref{eq3.2}.
Therefore the inner product $\vec{p}_L \cdot\vec{p}'_{L}$ is given by
\begin{equation}\label{eq3.17}
    \begin{split}
        \vec{p}_L \cdot\vec{p}'_{L}&=
        \left(\frac{I+B}{2}\gamma N +\gamma^{\ast}M\right)^{\top}
        \left(\frac{I+B}{2}\gamma N' +\gamma^{\ast}M'\right)\\
        &=\left( \frac{\gamma ^{\top}\gamma +\gamma^{\top} B\gamma}{2}N+M \right)^{\top}
        (\gamma^{\top}\gamma)^{-1}
        \left( \frac{\gamma ^{\top}\gamma +\gamma^{\top} B\gamma}{2}N'+M' \right)\\
        &=\frac{s}{2t}
        \left(\frac{t}{s}N+\frac{\epsilon N}{s}+M\right)^{\top}
        \left(\frac{t}{s}N'+\frac{\epsilon N'}{s}+M'\right)\\
        &=2st\left(\vec{m}+\frac{(2t,0,\dots ,0)^{\top}}{2st}\right)^{\top}
        \left(\vec{m}'+\frac{\vec{r}}{2st}\right)\\
        &=2st\cdot\frac{2t}{2st}\cdot\frac{r_1}{2st}\ \mod 1\\
        &=\frac{r_1}{s}\ \mod 1.
    \end{split}
\end{equation}
Of course this result is independent of the choice of $N,M,N',M',\vec{m},\vec{m}'$.
So one has the general monodromy charge $Q_i(\vec{r})=\frac{r_i}{s}$, and
this allows to write the monodromy charge in the vector form $\frac{\vec{r}}{s}$.

Then we have to compute the matrix $R_{ij}$, which is expected to be $\frac{2t}{s}I$.
We start with the $R_{11}$, which can be deduced from the conformal weight of 
$\mathcal{O}_{(2t,0,\dots ,0)}$. This corresponds to $(\alpha_1 ,\dots ,\alpha_n)
=(1,0,\dots ,0)$ in the L.H.S. of eq.\eqref{eq2.24}. 
Since $h(\mathcal{O}_{(2t,0,\dots ,0)})$ is equal to $\frac{\vec{p}_L \cdot\vec{p}_L}{2}$,
the equation \eqref{eq3.17} leads us to $h(\mathcal{O}_{(2t,0,\dots ,0)})=\frac{t}{s}$
where we set $r_1 =2t$ and divide the result by two.
Consequently, we get the equation for $r_{11}$ (Recall that $R_{ij}=\frac{r_{ij}}{p}$.)
\begin{equation}
    \frac{t}{s}=\frac{r_{11}}{2s} +\frac{r_{11}}{2}=\frac{(s+1)r_{11}}{2s}\ \mod 1,
\end{equation}
and one can easily verify the solution of this equation is $r_{11}=2t$.
Thus the all diagonal elements of $R_{ij}$ are given by $\frac{2t}{s}$.
Finally, we complete the proof by showing that the off-diagonal elements of $R_{ij}$ is 
always zero.
To this end, we will determine $R_{12}$ without loss of generality. 
One can verify that $h(\mathcal{O}_{(2t,2t,0,\dots ,0)})=h(\mathcal{O}_{(2t,0,\dots ,0)})
+h(\mathcal{O}_{(0,2t,0,\dots ,0)})=2\cdot\frac{t}{s}$ from the equation \eqref{eq3.17}.
Then the equation \eqref{eq2.24} where $(\alpha_1 ,\dots ,\alpha_n)
=(1,1,0,\dots ,0)$ becomes
\begin{equation}
    \frac{2t}{s}=\frac{(s+1)r_{11}}{2s}+\frac{(s+1)r_{11}}{2s}+R_{12}=\frac{t}{s}\times 2
    +R_{12}\  \mod 1
\end{equation}
as we have used $R_{ii}=\frac{r_{ii}}{s}=\frac{2t}{s}$ and $R_{ij}$ is symmetric.
This implies that $R_{12}=0$, and we can conclude that $R=\frac{2t}{s}I$.
Substituting $\frac{\vec{r}}{s}$ for $Q_i(\vec{r})$ in eq.\eqref{eq2.25} and 
$\frac{t}{s}I+\frac{\epsilon}{s}$ for $X=\frac{R}{2}+e$, we find the equation \eqref{eq3.13}
is equivalent to the twist equation of the simple current of $S$ with discrete torsion 
$\frac{\epsilon}{s}$. 
\end{proof}
With this lemma, we need to show that the sum over all $N,M$ of the equation \eqref{eq3.12}
exactly produces the irreducible character decomposition.
This means that we must exclude the case where 
$\frac{N}{2s}+\frac{\epsilon N}{2st}+\frac{M}{2t}=\frac{N'}{2s}+\frac{\epsilon N'}{2st}+
\frac{M'}{2t}=
\vec{m}+\frac{\vec{l}}{2st}$
for some $(N,M)\neq(N',M')$, and where for $(\vec{m},\vec{l})$ there is no $(N,M)$ such that 
$\frac{N}{2s}+\frac{\epsilon N}{2st}+\frac{M}{2t}=
\vec{m}+\frac{\vec{l}}{2st}$.

\begin{lem}\label{lem3.2}
    Let $\phi$ be a map from $\mathbb{Z}^{2n}$ to $(\frac{1}{2st}\mathbb{Z})^{2n}$ that maps 
    $\begin{pmatrix}
        N \\ M
    \end{pmatrix}$ to
    \begin{equation}
        \begin{pmatrix}
        \vec{m}+\frac{\vec{l}}{2st} \\
        \vec{m}'+\frac{-\vec{l}+2t\vec{a}}{2st}
    \end{pmatrix}
    \end{equation}
    where $(\vec{l},-\vec{a})$ is one of the solutions of the twist equation \eqref{eq3.13}, 
    and
    \begin{equation}
        \begin{split}
            \frac{N}{2s}+\frac{\epsilon N}{2st}+\frac{M}{2t}&=
            \vec{m}+\frac{\vec{l}}{2st}\\
            \frac{N}{2s}-\frac{\epsilon N}{2st}-\frac{M}{2t}&=
            \vec{m}'+\frac{-\vec{l}+2t\vec{a}}{2st}.
        \end{split}
    \end{equation}
    Then $\phi$ is injective, and 
    any $\begin{pmatrix}
        \vec{m}+\frac{\vec{l}}{2st} \\
        \vec{m}'+\frac{-\vec{l}+2t\vec{a}}{2st}
    \end{pmatrix}$
    for $\vec{m},\ \vec{m}'\in\mathbb{Z}^n$
    is always in the image of $\phi$
    if and only if $(\vec{l},-\vec{a})$ is one of the solutions of the twist equation
    \eqref{eq3.13}.
\end{lem}
\begin{proof}
    The map $\phi$ can be represented by the matrix 
    \begin{equation}
        A=
        \begin{pmatrix}
            \frac{I}{2s}+\frac{\epsilon}{2st} & \frac{I}{2t}\\
            \frac{I}{2s}-\frac{\epsilon}{2st} & -\frac{I}{2t}
        \end{pmatrix}.
    \end{equation}
    To prove the injectivity of $\phi$, there must be a map $\psi$ such that 
    $\psi\circ\phi =\id_{\mathbb{Z}^{2n}}$.
    This can be shown by introducing the inverse of $A$:
    \begin{equation}
        B=
        \begin{pmatrix}
            sI & sI \\
            tI+\epsilon^{\top} & -tI+\epsilon^{\top}
        \end{pmatrix}
    \end{equation}
    can be easily seen to be the inverse of $A$ 
    using that $\epsilon$ is antisymmetric, and this implies the injectivity
    of $\phi$.
    
    Suppose that $(\vec{l},-\vec{a})$ is a solution of the equation \eqref{eq3.13},
    and let $\vec{m},\vec{m}'$ be any integer vectors in $\mathbb{Z}^n$.
    Then, we get
    \begin{equation}
        B\begin{pmatrix}
            \vec{m}+\frac{\vec{l}}{2st} \\
        \vec{m}'+\frac{-\vec{l}+2t\vec{a}}{2st}
        \end{pmatrix}
        =
        \begin{pmatrix}
            s(\vec{m}+\vec{m}')+\vec{a} \\
            t(\vec{m}-\vec{m}')+\epsilon^{\top}(\vec{m}+\vec{m}')
            +\frac{\vec{l}}{s}+\left(\frac{t}{s}I+\frac{\epsilon}{s}\right)(-\vec{a})
        \end{pmatrix},
    \end{equation}
    whose second element of R.H.S. is an integer due to the twist equation.
    Therefore $\begin{pmatrix}
        \vec{m}+\frac{\vec{l}}{2st} \\
    \vec{m}'+\frac{-\vec{l}+2t\vec{a}}{2st}
    \end{pmatrix}$
    is in the image of $\phi$. The inverse can be deduced immediately 
    according to the lemma \ref{lem3.1}.
\end{proof}
As we proved the lemma \ref{lem3.1} and \ref{lem3.2}, 
we are ready to show the proposition \ref{prop3.1} as follows.
\begin{proof}[Proof of Prop.\ref{prop3.1}]
    The partition function of $\mathcal{A}$ is 
    \begin{equation}\label{eq3.25}
        \begin{split}
             \mathcal{Z}_{\mathcal{A}}(\tau,\bar{\tau})
        &=\frac{1}{|\eta(\tau)|^{2n}}
        \sum_{(\vec{p}_L,\vec{p}_R)\in\Lambda}
        q^{\frac{\|\vec{p}_L\|^2}{2}}\bar{q}^{\frac{\|\vec{p}_R\|^2}{2}}\\
        &=\frac{1}{|\eta(\tau)|^{2n}}
        \sum_{(N,M)\in\mathbb{Z}^{2n}}
        q^{st\left\| \frac{N}{2s} +\frac{\epsilon N}{2st} 
        +\frac{M}{2t} \right\|^2}
        \bar{q}^{st\left\| \frac{N}{2s} -\frac{\epsilon N}{2st} 
        -\frac{M}{2t} \right\|^2}.
        \end{split}
    \end{equation}
    In order to decompose the sum over $(N,M)$ into irreducible characters, we 
    consider the two sets $U,V$ such that 
    \begin{equation}
        \begin{split}
            U&=\left\{
                \left(
                    \frac{N}{2s} +\frac{\epsilon N}{2st} +\frac{M}{2t},
                    \frac{N}{2s} -\frac{\epsilon N}{2st} -\frac{M}{2t}
                \right)
                \mid
                (N,M)\in\mathbb{Z}^n \times\mathbb{Z}^n,
            \right\}\\
            V&=\left\{
                \left(
                    \vec{m}+\frac{\vec{l}}{2st},
                    \vec{m}'+\frac{-\vec{l}+2t\vec{a}}{2st}
                \right)    
                \mid
                \vec{m},\vec{m}'\in\mathbb{Z}^n,\ 
                (\vec{l},-\vec{a})\in \text{Sol}
            \right\},
        \end{split}
    \end{equation}
    where $\text{Sol}$ denotes the set of all solutions of the equation \eqref{eq3.13}.
    What the lemma \ref{lem3.1} and \ref{lem3.2} imply is that 
    $U$ and $V$ are identical to each other, not just an one-to-one correspondence.
    Subsequently, the equation \eqref{eq3.25} can be deformed into 
    \begin{equation}
        \frac{1}{|\eta(\tau)|^{2n}}
        \sum_{(\vec{l},\vec{a})\in\text{Sol}}
        \sum_{\vec{m},\vec{m}'\in\mathbb{Z}^n}
        q^{st\left\|\vec{m}+\frac{\vec{l}}{2st}\right\|^2}
        \bar{q}^{st\left\|\vec{m}'+\frac{-\vec{l}-2t\vec{a}}{2st}\right\|^2}.
    \end{equation}
    For convenience, here we have substituted $\vec{a}$ for $-\vec{a}$.
    Referring to the equation \eqref{eq3.10}, we have
    \begin{equation}
        \mathcal{Z}_{\mathcal{A}}(\tau,\bar{\tau})
        =\sum_{(\vec{l},\vec{a})\in\text{Sol}}
        \chi^{(2st)}_{\vec{l}}\bar{\chi}^{(2st)}_{-\vec{l}-2t\vec{a}},
    \end{equation}
    which means that $\mathcal{A}$ is the simple current orbifold of $B=0$ case 
    with discrete torsion $\frac{\epsilon}{p}$.
\end{proof}
The proposition \ref{prop3.1} enables us to easily determine the combination
of characters of primary fields by means of the discrete torsion, 
which can be calculated with $\gamma$ and $B$.
In addition, we can construct an error-correcting code 
from $\mathcal{A}$, which can reproduce the partition function with use of 
the weight enumerator polynomial. This might be the inverse of Construction A, and 
partially the extension of the work of Buican et al. \cite{buican2021}.
Before describing this, we will continue to illustrate how to regard a Narain CFT as
a simple current orbifold model for non-homogeneous cases.
\subsubsection{General Case}\label{sec3.1.2}
Similarly to the homogeneous case, we first calculate the Verlinde algebra for $B=0$ in
the general case.
The momentum vectors $\frac{\|\vec{p}_L\|^2}{2}$, $\frac{\|\vec{p}_R\|^2}{2}$
can be found to 
\begin{equation}
    \begin{split}
        \frac{\|\vec{p}_L\|^2}{2}&=\sum_{i=1}^{n}
        s_i t_i \left( \frac{N_i}{2s_i} +\frac{M_i}{2t_i} \right)^2\\
        \frac{\|\vec{p}_R\|^2}{2}&=\sum_{i=1}^{n}
        s_i t_i \left( \frac{N_i}{2s_i} -\frac{M_i}{2t_i} \right)^2.
    \end{split}
\end{equation}
Then it is obvious that the partition function can be written as 
\begin{equation}\label{eq3.30}
    \begin{split}
        &\mathcal{Z}(\tau ,\bar{\tau})\\
        &=
    \frac{1}{|\eta(\tau)|^2}
    \sum_{\vec{l}=(l_1,\dots ,l_n)\in K}
    \sum_{\substack{(m_1,\dots ,m_n)\in\mathbb{Z}^n\\(m'_1,\dots ,m'_n)\in\mathbb{Z}^n}}
    \prod_{i=1}^{n}q^{s_i t_i \left(m_i+\frac{l_i}{2s_i t_i}\right)^2}
    \bar{q}^{s_i t_i \left(m_i ' +\frac{(2t_i -1)l_i}{2s_i t_i}\right)^2}\\
    &=
    \sum_{\vec{l}\in K}
    \chi ^{(2s\cdot t)}_{\vec{l}}\bar{\chi}^{(2s\cdot t)}_{\vec{l}'},\ 
    l'_i = (2t_i -1)l_i, 
    \end{split}
\end{equation}
where $K=\mathbb{Z}_{2s_1 t_1}\times\cdots\times\mathbb{Z}_{2s_n t_n}$, and 
$\chi^{(2s\cdot t)}_{\vec{l}}=\prod_{i=1}^n \chi^{(2s_i t_i)}_{l_i}$.
Therefore, the Verlinde algebra for $B=0$ is $K$.
Next, we consider the set of simple currents $S'$ generated by 
$\{\mathcal{O}_{(2t_1,0,\dots,0)},\mathcal{O}_{(0,2t_2,0,\dots,0)},\dots,
\mathcal{O}_{(0,\dots,0,2t_n)}\}$, which forms the $\mathbb{Z}_{s_1}\times\dots\times
\mathbb{Z}_{s_n}$ subgroup of $K$.
Now we are going to prove one of our main theorem where the homogeneousness 
in the proposition \ref{prop3.1} is omitted.
\begin{thm}\label{thm3.1}
    Let $\mathcal{A}'$ be a Narain CFT with its compactifying lattice satisfying the 
    assumption \ref{asmp3.1}. 
    If the following equation 
    \begin{equation}\label{eq3.31}
        \left(\frac{\gamma^{\top}B\gamma}{2}\right)_{ij}=\frac{\epsilon _{ij}}{s_i}
    \end{equation}
    holds for some $\epsilon _{ij}\in\mathbb{Z}$, 
    then $\mathcal{A}'$ is identical to the simple current orbifold $S'$ of 
    $B=0$ case \eqref{eq3.30}, whose discrete torsion is 
    $e_{ij}=\frac{\epsilon _{ij}}{s_i}$.
\end{thm}
Note that in the non-homogeneous case, the matrix $(\epsilon)_{ij}=\epsilon_{ij}$
is no longer antisymmetric, but $\frac{\epsilon_{ij}}{s_i}$ is, since 
$\frac{\gamma^{\top}B\gamma}{2}$ is antisymmetric.
However, this can be proved in an analogous way to the proof of the 
proposition \ref{prop3.1},
focusing on each component.
\begin{proof}
The explicit form of $\frac{\|\vec{p}_L\|^2}{2}$ and $\frac{\|\vec{p}_R\|^2}{2}$
is 
\begin{equation}\label{eq3.32}
    \begin{split}
        \frac{\|\vec{p}_L\|^2}{2}&=\sum_{i=1}^n 
        s_i t_i\left( \frac{N_i}{2s_i} +\frac{\epsilon_{ij} N_j}{2s_i t_i} 
        +\frac{M_i}{2t_i} \right)^2 \\
        \frac{\|\vec{p}_R\|^2}{2}&=\sum_{i=1}^n
        s_i t_i\left( \frac{N_i}{2s_i} -\frac{\epsilon_{ij} N_j}{2s_i t_i} 
        -\frac{M_i}{2t_i} \right)^2
    \end{split}
\end{equation}
where we abbreviated $\sum_{j}$.
Then if $\vec{p}_L$ has the label $\vec{l}\in K$, $\vec{p}_R$ has 
$-\vec{l}+2\vec{t}\cdot N$ for $\vec{t}=(t_1,\dots,t_n)$.
Besides, the pair $(\vec{l},-N)$ is found to be a solution of 
\begin{equation}\label{eq3.33}
    \frac{r_i}{s_i}+\left(\frac{t_i}{s_i}\delta_{ij}+\frac{\epsilon_{ij}}{s_i}\right)
    \beta_j =0\ \mod 1,\ \ \forall i 
\end{equation}
for $(\vec{r}=(r_1,\dots,r_n),\vec{\beta})$.
A similar calculation to the proof of the lemma \ref{lem3.1} brings us to 
that \eqref{eq3.33} is equivalent to the twist equation of the simple current orbifold 
of $S'$ with discrete torsion $\frac{\epsilon_{ij}}{p_i}$.
This can be seen by referring to the equation \eqref{eq3.17},
\begin{equation}
    \begin{split}
        \vec{p}_L \cdot\vec{p}'_{L}&=
        \left(\frac{I+B}{2}\gamma N +\gamma^{\ast}M\right)^{\top}
        \left(\frac{I+B}{2}\gamma N' +\gamma^{\ast}M'\right)\\
        &=\left( \frac{\gamma ^{\top}\gamma +\gamma^{\top} B\gamma}{2}N+M \right)^{\top}
        (\gamma^{\top}\gamma)^{-1}
        \left( \frac{\gamma ^{\top}\gamma +\gamma^{\top} B\gamma}{2}N'+M' \right)\\
        &=\sum_{i=1}^n \frac{s_i}{2t_i}
        \left(\frac{t_i}{s_i}N_i+\frac{\epsilon_{ij} N_j}{s_i}+M_i\right)
        \left(\frac{t_i}{s_i}N'_i +\frac{\epsilon_{ij} N'_j}{s_i}+M'_i\right).
    \end{split}
\end{equation}
The monodromy charge $Q_i(\vec{r})$ and the matrix $R$ can be found easily from the 
above equation. That is, $Q_i(\vec{r})=\frac{r_i}{s_i}$ and $R_{ij}=\frac{t_i}{s_i}\delta_{ij}$.

The last thing we have to confirm is that
            
                \begin{equation}
                    \begin{split}
                        U'&=\left\{
                            \left(
                        \begin{pmatrix}
                            \frac{N_1}{2s_1} +\frac{\epsilon_{1j} N_j}{2s_1 t_1} +\frac{M_1}{2t_1} \\
                            \vdots \\
                            \frac{N_n}{2s_n} +\frac{\epsilon_{nj} N_j}{2s_n t_n} +\frac{M_n}{2t_n} 
                        \end{pmatrix},
                        \begin{pmatrix}
                            \frac{N_1}{2s_1} -\frac{\epsilon_{1j} N_j}{2s_1 t_1} -\frac{M_1}{2s_1}\\
                            \vdots \\
                            \frac{N_1}{2s_1} -\frac{\epsilon_{1j} N_j}{2s_1 t_1} -\frac{M_1}{2t_1}
                        \end{pmatrix}
                        \right)
                            \mid
                            (N,M)\in\mathbb{Z}^n \times\mathbb{Z}^n,
                        \right\},\\
                        V'&=\left\{
                            \left(
                                \begin{pmatrix}
                                    m_1 +\frac{l_1}{2s_1 t_1}\\
                                    \vdots \\
                                    m_n +\frac{l_n}{2s_n t_n}
                                \end{pmatrix},
                                \begin{pmatrix}
                                    m'_1 +\frac{-l_1 +2t_1 a_1}{2s_1 t_1}\\
                                    \vdots \\
                                    m'_n +\frac{-l_n +2t_n a_n}{2s_n t_n}
                                \end{pmatrix}
                            \right)    
                            \mid
                            \vec{m},\vec{m}'\in\mathbb{Z}^n,\ 
                            (\vec{l},-\vec{a})\in \text{Sol}'
                        \right\}
                    \end{split}
                \end{equation}
are identical as sets. Here we denote by $\text{Sol}'$ the set of 
all solutions of the equation \eqref{eq3.33}.
Consider a map $\phi':\mathbb{Z}^n \times\mathbb{Z}^n \to (\frac{1}{2s_1 t_1}\mathbb{Z}
\times\cdots\times\frac{1}{2s_n t_n}\mathbb{Z})^2$ given by
\begin{equation}
    \begin{pmatrix}
        N \\ M
    \end{pmatrix}
    \mapsto
    \left(
        \begin{pmatrix}
            \frac{N_1}{2s_1} +\frac{\epsilon_{1j} N_j}{2s_1 t_1} +\frac{M_1}{2t_1} \\
            \vdots \\
            \frac{N_n}{2s_n} +\frac{\epsilon_{nj} N_j}{2s_n t_n} +\frac{M_n}{2t_n} 
        \end{pmatrix},
        \begin{pmatrix}
            \frac{N_1}{2s_1} -\frac{\epsilon_{1j} N_j}{2s_1 t_1} -\frac{M_1}{2t_1}\\
            \vdots \\
            \frac{N_1}{2s_1} -\frac{\epsilon_{1j} N_j}{2s_1 t_1} -\frac{M_1}{2t_1}
        \end{pmatrix}
    \right).
\end{equation}
This map can be regarded as a matrix 
$A'=\begin{pmatrix}
        A_1 & A_2 \\ A_3 & A_4
    \end{pmatrix}$
which acts on $\begin{pmatrix}
    N\\M
\end{pmatrix}$,
where
\begin{equation}
    \begin{split}
    (A_1)_{ij}&=\frac{\delta_{ij}}{2s_i} +\frac{\epsilon_{ij}}{2s_i t_i}
    ,\ 
    (A_2)_{ij}=\frac{\delta_{ij}}{2t_i},
    \\
    (A_3)_{ij}&=\frac{\delta_{ij}}{2s_i} -\frac{\epsilon_{ij}}{2s_i t_i}
    ,\ 
    (A_4)_{ij}=-\frac{\delta_{ij}}{2t_i}.
    \end{split}
\end{equation}
Comparably to the lemma \ref{lem3.2}, one can find the inverse matrix of $A'$ which
we denote by 
$B'=\begin{pmatrix}
    B_1 & B_2 \\ B_3 & B_4
\end{pmatrix}$
where 
\begin{equation}
    \begin{split}
        (B_1)_{ij}=s_i \delta_{ij}
        &,\ 
        (B_2)_{ij}=s_i \delta_{ij},
        \\
        (B_3)_{ij}=s_i \delta_{ij}+\epsilon_{ji}
        &,\ 
        (B_4)_{ij}=-t_i \delta_{ij}+\epsilon_{ji}.
        \end{split}
\end{equation}
Then it is easy to see that $\phi'$ is injective, and that 
$$\left(
    \begin{pmatrix}
        m_1 +\frac{l_1}{2s_1 t_1}\\
        \vdots \\
        m_n +\frac{l_n}{2s_n t_n}
    \end{pmatrix},
    \begin{pmatrix}
        m'_1 +\frac{-l_1 +2t_1 a_1}{2s_1 t_1}\\
        \vdots \\
        m'_n +\frac{-l_n +2t_n a_n}{2s_n t_n}
    \end{pmatrix}
\right)$$
 for $\vec{m},\ \vec{m}'\in\mathbb{Z}^n$ is always in the image of 
$\phi'$ if and only if $(\vec{l},-\vec{a})$ satisfy the equation \eqref{eq3.33}. 
Therefore, the partition function of $\mathcal{A}'$ becomes 
\begin{equation}\label{eq3.39}
    \mathcal{Z}_{\mathcal{A}'}(\tau,\bar{\tau})=
    \sum_{(\vec{l},\vec{a})\in\text{Sol}'}
    \chi^{(2s\cdot t)}_{\vec{l}}\bar{\chi}^{(2s\cdot t)}_{-\vec{l}-2\vec{t}\cdot\vec{a}},
\end{equation}
which completes the proof. 
\end{proof}
Then the theorem \ref{thm3.1} gives the sufficient condition for
a Narain CFT to be rational. In addition, the relation between the antisymmetric
B-field and the discrete torsion of the simple current orbifold is 
provided.
Moreover, in the next subsection we will show that these Narain RCFTs have ``code-like'' structures 
in general.
\subsection{Code Construction}\label{sec3.2}
It has been argued that the partition function of code CFT can be 
described with the enumerator polynomial of the code whose variables 
are substituted for certain $q$-series. Let us put a code $\mathcal{C}$ over $\mathbb{F}_p$ 
with length $2n$ as 
an example. 
Each codeword $c\in\mathcal{C}$ can be represented in terms of $n$-elements vectors 
$\alpha,\ \beta\in\mathbb{F}_p ^n$ such that 
$c=(\alpha ,\beta)$.
The enumerator polynomial $W_{\mathcal{C}}(\{x_{ab}\})$ is defined as \cite{nebe2006a}
\begin{equation}\label{eq3.40}
    W_{\mathcal{C}}(\{x_{ab}\})=\sum_{c\in\mathcal{C}}\prod_{(a,b)\in\mathbb{F}_p \times\mathbb{F}_p}
    x_{ab}^{\text{wt}_{ab}(c)}
\end{equation}
where 
\begin{equation}
    \text{wt}_{ab}(c)=\sharp\{i\mid (\alpha_i,\beta_i)=(a,b)\}.
\end{equation}
Then the partition function of the code CFT associated with $\mathcal{C}$ can be written as 
\cite{kawabata2022}
\begin{equation}
    \mathcal{Z}_{\mathcal{C}}(\tau,\bar{\tau})
    =\frac{1}{|\eta(\tau)|^{2n}}W_{\mathcal{C}}(\{\psi_{ab}\}),
\end{equation}
the variables of $W_{\mathcal{C}}$ being replaced by 
\begin{equation}
    \psi^{(p)}_{ab}=\sum_{k_1 ,k_2\in\mathbb{Z}}
    q^{\frac{p}{4}\left(k_1+k_2+\frac{a+b}{p}\right)^2}
    \bar{q}^{\frac{p}{4}\left(k_1-k_2+\frac{a-b}{p}\right)^2}.
\end{equation}
Of course, this is consistent with the $p=2$ case \cite{dymarsky2020a}, where 
the code is viewed as a quantum stabilizer code.
The second main result in our paper states that the partition function 
\eqref{eq3.39} has the structure of the enumerator polynomial of a code-like 
object, and we will show this in the following.
\subsubsection{Construction}\label{sec3.2.1}
Similarly to the section \ref{sec3.1}, we start with the homogeneous case, 
and we set $st>1$.
Recall that the irreducible character with the label $\vec{l}$ is 
\begin{equation}
    \chi^{(2st)}_{\vec{l}}=\prod_{i=1}^{n}\chi^{(2st)}_{l_i}
    =\frac{1}{|\eta(\tau)|^{2n}}
    \prod_{i=1}^{n}\sum_{m_i\in\mathbb{Z}}q^{st\left(m_i+\frac{l_i}{2st}\right)^2}.
\end{equation}
We are going to show that combining these characters appropriately produces 
the structure of a enumerator polynomial.
By $\varphi^{(d)}_{ij}$ we denote the $q$-series 
\begin{equation}
    \sum_{k_1,k_2\in\mathbb{Z}}q^{\frac{d}{4}\left(k_1+k_2+\frac{i}{d}\right)^2}
    \bar{q}^{\frac{d}{4}\left(k_1-k_2+\frac{j}{d}\right)}.
\end{equation}
Then one finds the following lemma:
\begin{lem}\label{lem3.3}
    The following equation holds:
\begin{equation}\label{eq3.46}
    \sum_{\substack{\vec{u}=(u_1,\dots,u_n)\\u_i\in\{0,st\}}}
    \chi^{(2st)}_{\vec{l}+\vec{u}}\bar{\chi}^{(2st)}_{\vec{l}'+\vec{u}}
    =\frac{1}{|\eta(\tau)|^{2n}}\prod_{i=1}^{n}\varphi^{(st)}_{l_i l'_i}.
\end{equation}

\end{lem}
\begin{proof}
    This can be shown by a simple calculation.
    First, consider the $n=1$ case.
\begin{equation}\label{eq3.47}
    \begin{split}
        \chi^{(2st)}_{l}\bar{\chi}^{(2st)}_{l'}
        +\chi^{(2st)}_{l+st}\bar{\chi}^{(2st)}_{l'+st}
        &=\frac{1}{|\eta(\tau)|^{2}}
        \sum_{m_1,m_2\in\mathbb{Z}}q^{st\left(m_1+\frac{l}{2st}\right)^2}
        \bar{q}^{st\left(m_2+\frac{l'}{2st}\right)^2}
        +\\
        &\frac{1}{|\eta(\tau)|^{2}}
        \sum_{m'_1,m'_2\in\mathbb{Z}}q^{st\left(m'_1+\frac{l+st}{2st}\right)^2}
        \bar{q}^{st\left(m'_2+\frac{l'+st}{2st}\right)^2}\\
        &=\frac{1}{|\eta(\tau)|^{2}} 
        \sum_{m_1,m_2\in\mathbb{Z}}q^{\frac{st}{4}\left(2m_1+\frac{l}{st}\right)^2}
        \bar{q}^{\frac{st}{4}\left(2m_2+\frac{l'}{st}\right)^2}
        +\\
        &\frac{1}{|\eta(\tau)|^{2}}
        \sum_{m'_1,m'_2\in\mathbb{Z}}q^{\frac{st}{4}\left(2m'_1 +1+\frac{l}{st}\right)^2}
        \bar{q}^{\frac{st}{4}\left(2m'_2 +1+\frac{l'}{st}\right)^2}\\
        &=\frac{1}{|\eta(\tau)|^{2}}\sum_{(m_1,m_2):\text{same parity}}
        q^{\frac{st}{4}\left(m_1+\frac{l}{st}\right)^2}
        \bar{q}^{\frac{st}{4}\left(m_2+\frac{l'}{st}\right)}\\
        &=\frac{1}{|\eta(\tau)|^{2}}\sum_{k_1,k_2\in\mathbb{Z}}
        q^{\frac{st}{4}\left(k_1+k_2+\frac{l}{st}\right)^2}
        \bar{q}^{\frac{st}{4}\left(k_1-k_2+\frac{l'}{st}\right)}\\
        &=\frac{1}{|\eta(\tau)|^{2}}\varphi^{(st)}_{l, l'}.
    \end{split}
\end{equation}
Then the left hand side of the equation \eqref{eq3.46} can be 
factorized into 
\begin{equation}
    \left(\chi^{(2st)}_{l_1}\bar{\chi}^{(2st)}_{l'_1}+
    \chi^{(2st)}_{l_1+st}\bar{\chi}^{(2st)}_{l'_1+st}\right)
    \cdots
    \left(
        \chi^{(2st)}_{l_n}\bar{\chi}^{(2st)}_{l'_n}+
    \chi^{(2st)}_{l_n+st}\bar{\chi}^{(2st)}_{l'_n+st}
    \right),
\end{equation}
which equals to the right hand side of the equation \eqref{eq3.46} due to 
the equation \eqref{eq3.47}.
\end{proof}
In order to use this we need the following lemma, but it is 
self-evident by the equation \eqref{eq3.13}.
\begin{lem}\label{lem3.4}
    Let $\mathcal{A}$ be a Narain CFT that fulfill the assumption of 
    the proposition \ref{prop3.1}.
    Then if $(\vec{l},\vec{a})\in\text{Sol}$, 
    \begin{equation}
        \{(\vec{l}+\vec{u},\vec{a})\mid \vec{u}=(u_1,\dots,u_n),\ u_i\in\{0,st\}\}
        \subset\text{Sol}.
    \end{equation}
\end{lem} 
We also need to define the weight enumerator polynomial of a additive subgroup of 
$(\mathbb{Z}_d)^n \times(\mathbb{Z}_d)^n$.
\begin{dfn}
Let $\mathcal{C}$ be an additive subgroup of $(\mathbb{Z}_d)^n \times(\mathbb{Z}_d)^n$ and 
suppose that each $c\in\mathcal{C}$ can be written as $c=(\alpha,\beta)$ with 
$\alpha,\beta\in\mathbb{Z}_d^n$.
Then the enumerator polynomial of $\mathcal{C}$ is defined as follows.
\begin{equation}
    W_{\mathcal{C}}(\{x_{ab}\})=\sum_{c\in\mathcal{C}}\prod_{(a,b)\in\mathbb{Z}_d \times\mathbb{Z}_d}
    x_{ab}^{\text{wt}_{ab}(c)}
\end{equation}
where 
\begin{equation}
    \text{wt}_{ab}(c)=\sharp\{i\mid(a,b)=c_i\}.
\end{equation}
\end{dfn}
Finally, we are ready to construct the code from primary fields,
which is described in this theorem.
\begin{prop}\label{prop3.2}
    Following the notation of \cite{buican2021},
    we define the map $\mu$ from the set of primary fields of $\mathcal{A}$ to 
    $(\mathbb{Z}_{st})^n \times(\mathbb{Z}_{st})^n$ as
    \begin{equation}
        \mu :\mathcal{O}_{\vec{l},-\vec{l}-2t\vec{a}}\mapsto
        (-t\vec{a},\vec{l}+t\vec{a})\ \mod st.
    \end{equation}
    
    Let $\mathcal{C}(\mathcal{A})$ be the image of $\mu$.
    Then the partition function of $\mathcal{A}$ can be written as 
    \begin{equation}
        \mathcal{Z}_{\mathcal{A}}(\tau,\bar{\tau})=\frac{1}{|\eta(\tau)|^{2n}}
        W_{\mathcal{C}(\mathcal{A})}(\{\psi^{(st)}_{ab}\})
    \end{equation}
\end{prop}
\begin{proof}
    A simple calculation leads us to 
    \begin{equation}
        \varphi^{(st)}_{i,-i-2tj}=\psi^{(st)}_{-tj,i+tj}.
    \end{equation}
    We introduce the equivalence relation in $\text{Sol}$ as 
    \begin{equation}
        (\vec{l},\vec{a})\sim(\vec{l}',\vec{a})\Leftrightarrow
        \vec{l}-\vec{l}'\in\{0,\pm st\}^n.
    \end{equation}
    Then according to the lemma \ref{lem3.3},\ref{lem3.4}, one has
    \begin{equation}
        \begin{split}
            \mathcal{Z}_{\mathcal{A}}(\tau,\bar{\tau})&=
            \sum_{(\vec{l},\vec{a})\in\text{Sol}}
            \chi^{(2st)}_{\vec{l}}\bar{\chi}^{(2st)}_{-\vec{l}-2t\vec{a}}\\
            &=\sum_{(\vec{l},\vec{a})\in\text{Sol}/\sim}
            \sum_{\substack{\vec{u}=(u_1,\dots,u_n)\\u_i\in\{0,st\}}}
            \chi^{(2st)}_{\vec{l}+\vec{u}}\bar{\chi}^{(2st)}_{-\vec{l}-2t\vec{a}+\vec{u}}\\
            &=\frac{1}{|\eta(\tau)|^{2n}}
            \sum_{(\vec{l},\vec{a})\in\text{Sol}/\sim}
            \prod_{i=1}^{n}\varphi^{(st)}_{l_i, -l_i-2ta_i}\\
            &=\frac{1}{|\eta(\tau)|^{2n}}
            \sum_{(\vec{l},\vec{a})\in\text{Sol}/\sim}
            \prod_{i=1}^{n}\psi^{(st)}_{-ta_i,l_i+ta_i}\\
            &=\frac{1}{|\eta(\tau)|^{2n}}
            \sum_{(\vec{l},\vec{a})\in\text{Sol}/\sim}
            \prod_{(a,b)\in\mathbb{Z}_{st}\times\mathbb{Z}_{st}}
            (\psi_{ab}^{(st)})^{\text{wt}_{ab}(-t\vec{a},\vec{l}+t\vec{a})}\\
            &=\frac{1}{|\eta(\tau)|^{2n}}
            \sum_{c\in\mathcal{C}(\mathcal{A})}
            \prod_{(a,b)\in\mathbb{Z}_{st}\times\mathbb{Z}_{st}}
            (\psi_{ab}^{(st)})^{\text{wt}_{ab}(c)},
        \end{split}
    \end{equation}
    and the last term is equal to $\frac{1}{|\eta(\tau)|^{2n}}
    W_{\mathcal{C}(\mathcal{A})}(\{\psi^{(st)}_{ab}\})$.
\end{proof}
It is obvious that for the case $s=2,t=1$ the definition of $\mu$ reduces to
the one of \cite{buican2021}. 
Before discussing some examples, we should verify that this code construction can be done 
even for the non-homogeneous case.
In this case, the code structure has different dimensions on each bit.
Therefore we will call it a code 
an additive subgroup of $(\mathbb{Z}_{d_1}\times\dots\times\mathbb{Z}_{d_n})^2$.

\begin{dfn}
    Let $\mathcal{C}'$ be an additive subgroup of 
    $(\mathbb{Z}_{d_1}\times\dots\times\mathbb{Z}_{d_n})^2$, and suppose that 
    each $c\in\mathcal{C}'$ can be written as $c=(\alpha,\beta)$ with 
    $\alpha,\beta\in\mathbb{Z}_{d_1}\times\dots\times\mathbb{Z}_{d_n}$.
    Then the enumerator polynomial of $\mathcal{C}'$ is defined as 
    \begin{equation}
        W_{\mathcal{C}'}(\{x_{1,a_1 b_1}\},\dots,\{x_{n,a_n b_n}\})=\sum_{c\in\mathcal{C}'}
        \prod_{i=1}^{n}\prod_{(a_i,b_i)\in\mathbb{Z}_{d_i}\times\mathbb{Z}_{d_i}}
        x_{i,a_i b_i}^{\iota_{i,a_i b_i}(c)}
    \end{equation}
    where 
    \begin{equation}
        \iota_{i,a_i b_i}(c)=
        \begin{cases}
            1\ (\alpha_i,\beta_i)=(a_i ,b_i)\in\mathbb{Z}_{d_i}\times\mathbb{Z}_{d_i} \\
            0\ \text{otherwise}
        \end{cases}.
    \end{equation}

\end{dfn}
Extending the lemma \ref{lem3.3} to this case yields
\begin{equation}\label{eq3.59}
    \sum_{\substack{\vec{u}=(u_1,\dots,u_n)\\u_i\in\{0,s_i t_i\}}}
    \chi^{(2s\cdot t)}_{\vec{l}+\vec{u}}\bar{\chi}^{(2s\cdot t)}_{\vec{l}'+\vec{u}}
    =\frac{1}{|\eta(\tau)|^{2n}}\prod_{i=1}^{n}\varphi^{(s_i t_i)}_{l_i l'_i}.
\end{equation}
Then with analogous argument of the homogeneous case, one finds the following theorem.
\begin{thm}
    Let $\mathcal{A}'$ be a Narain CFT satisfying the assumption of the proposition 
    \ref{thm3.1}.
    We define the map $\mu '$ from the set of primary fields of $\mathcal{A}'$ to 
    $(\mathbb{Z}_{s_1 t_1}\times\cdots\times\mathbb{Z}_{s_n t_n})^2$ as
    \begin{equation}
        \mu' :\mathcal{O}_{\vec{l},-\vec{l}-2\vec{t}\cdot\vec{a}}\mapsto
        (-t_1 a_1,\dots,-t_n a_n,l_1+t_1 a_1,\dots,l_n+t_n a_n).
    \end{equation}
    where each $i$-th and $i+st$-th component is taken by $\mod s_i t_i$.
    Let $\mathcal{C}(\mathcal{A}')$ be the image of $\mu'$.
    Then the partition function of $\mathcal{A}'$ can be written as 
    \begin{equation}
        \mathcal{Z}_{\mathcal{A}'}(\tau,\bar{\tau})=\frac{1}{|\eta(\tau)|^{2n}}
        W_{\mathcal{C}(\mathcal{A}')}
        (\{\psi^{(s_1 t_1)}_{a_1 b_1}\},\dots,\{\psi^{(s_n t_n)}_{a_n b_n}\}).
    \end{equation}
\end{thm}
The proof of this will be skipped, as the proof can be 
proceeded by simply extending the one of the proposition \ref{prop3.2} 
to the component-wise argument.

\subsubsection{Some Examples}\label{sec3.2.2}
Now we are going to indicate some important examples.
By actually checking these examples, we can realize 
how discrete torsion is reflected in code theory. 
This fact motivates us to investigate how the discrete torsion preserves the 
information of the spectral gap, but we do not discuss that far in this paper.

\subsubsection*{$\mathbb{F}_p$ code}
Let $\mathcal{A}$ be a homogeneous Narain CFT with discrete torsion being described as 
the equation \eqref{eq3.11}. 
First, we consider the case $s=p$, $t=1$ where $p$ is a prime number.
We can take $\gamma=\sqrt{\frac{2}{p}}I$ and $B=\epsilon$ with an antisymmetric B-field 
$\epsilon$, for instance.
In this case, the image of $\mu$, which we denote by $\mathcal{C}(\mathcal{A})$, 
is 
\begin{equation}
    \mathcal{C}(\mathcal{A})=\{(-\vec{a},\vec{l}+\vec{a})\mod p\mid(\vec{l},\vec{a})
    \in\text{Sol}\}.
\end{equation}
One can verify that this is a B-form $\mathbb{F}_p$ code.
First, recall that for every $\vec{a}$, there are some $\vec{l}$ such that 
$(\vec{l},\vec{a})\in\text{Sol}$.
Moreover, from the twist equation
\begin{equation}
    \frac{\vec{l}}{p}+\left(\frac{I}{p}+\frac{\epsilon}{p}\right)\vec{a}=0\mod 1,
\end{equation}
one has 
\begin{equation}
    \vec{l}+\vec{a}=-\epsilon \vec{a}\mod p.
\end{equation}
Here note that the vectors $\vec{l},\vec{a}$ are column vectors, and that
in order to compare $\mathcal{C}(\mathcal{A})$ to a $\mathbb{F}_p$ code,
we should take transpose these vectors.
Let $\mathcal{C}(\mathcal{A})^{\top}$ be the set of vectors after 
transposing $(-\vec{a},\vec{l}+\vec{a})\in\mathcal{C}(\mathcal{A})$.
Then each element of $\mathcal{C}(\mathcal{A})^{\top}$ can be written as 
$(-\vec{a}^{\top},-\vec{a}^{\top}\epsilon^{\top})$ for $\vec{a}\in\mathbb{Z}_p^ n$.
On the other hand, a codeword $c$ of a $\mathbb{F}_p$ code $\mathcal{C}(G)$ with length $2n$ and 
dimension $n$ that has the generator matrix $G=
    \begin{pmatrix}
        A & B
    \end{pmatrix}$
is of form $(\vec{j}A,\vec{j}B)$ where $\vec{j}\in\mathbb{Z}_p ^n$ is a 
row vector.
Comparing to $\mathcal{C}(\mathcal{A})^{\top}$, one can deduce that 
$\mathcal{C}(\mathcal{A})^{\top}$ is identical to the code over $\mathbb{F}_p$ with 
the generator matrix
\begin{equation}
    G'=
    \begin{pmatrix}
        -I & -\epsilon^{\top}
    \end{pmatrix}.
\end{equation}
In addition, $G'$ is equivalent to the B-form code 
$
    \begin{pmatrix}
        I & \epsilon^{\top}
    \end{pmatrix}.
$
This is consistent with that the code CFT constructed from the B-form code 
$\begin{pmatrix}
    I & B^{\top}
\end{pmatrix}$
is 
\begin{equation}
    \Lambda=\begin{pmatrix}
        pI & B \\
        0 & I
    \end{pmatrix}/\sqrt{p},
\end{equation}
where $\frac{\gamma^{\top}B\gamma}{2}=\frac{B}{p}$.
For relating the spectral gap to the code quantity, we introduce the squared Lee distance
$d(\mathcal{C})$,
\begin{equation}
d(\mathcal{C})\coloneqq\min_{0\neq c\in\mathcal{C}}\left\{ \frac{1}{p}
\sum_{i=1}^{2n}\min\{c_i^2,(p-c_i)^2\} 
\right\}.
\end{equation} 
Therefore, the spectral gap of $\mathcal{A}$ can be described with $d(\mathcal{C})$
of $\mathcal{C}(G)$ \cite{angelinos2022},
\begin{equation}
    \Delta=\frac{1}{2}\min\{d(\mathcal{C}),p\}.
\end{equation}

\subsubsection*{$\mathbb{Z}_d$ code}
Let us consider the case where $s=d$ is not a prime, but just an integer.
In the above discussion we did not use that $p$ is a prime number, 
so the extension to the general integer case can be easily carried out 
by substituting $p$ with $d$.
Then the equality between the discrete torsion ond the B-form can be proved.
However, the notion of the resulting code is changed because $\mathbb{Z}_d$ for 
non prime $d$ cannot have the structure of a finite field.
It can be resolved by redefining the code as an additive subgroup of 
$\mathbb{Z}_d^n \times\mathbb{Z}_d^n$. 

A remarkable example for $\mathbb{Z}_d$ code is the case where $t\neq 1$.
The simplest case is the $n=1$ case with $\gamma=(\sqrt{3})$,
where $s=2$, $t=3$.
Solving the twist equation, one finds that $K=\mathbb{Z}_{12}$ and that 
the set of primary fields is 
$\{\mathcal{O}_{(i,5i)}\mid i\in\mathbb{Z}_{12}\}$.
Then calculating the image of $\mu$ leads to 
\begin{equation}
    \im \mu =\{(0,0),(3,2),(0,4),(3,0),(0,2),(3,4)\}\subset\mathbb{Z}_6 \times\mathbb{Z}_6,
\end{equation}
which is generated by $(3,2)$ with the addition.
Note that in $n=1$ the antisymmetric B-field vanishes, and so as the discrete torsion.
Therefore, the latter $\mathbb{Z}_6$ component of this code is not related 
to the discrete torsion, in contrast to the $\mathbb{F}_p$ code.
This is because the twist equation is equivalent to an integer equation taken by $\mod s$,
but each component of $\im\mu$ is $\mod st$.
Then an equality in $\mod s$ does not mean an equality in $\mod st$, 
one cannot proceed the argument in the previous example.

\subsubsection*{$\mathbb{Z}_2 \times \mathbb{Z}_6$ code}
The third example is the non-homogeneous case.
We will consider the Narain CFT $\mathcal{A}'$ with 
\begin{equation}
    \gamma=\begin{pmatrix}
        1 & 0 \\ 0 & \frac{1}{3}
    \end{pmatrix},\ 
    B=\begin{pmatrix}
        0 & \sqrt{3} \\ -\sqrt{3} & 0
    \end{pmatrix},
\end{equation}
where $s_1=2,\ t_1=1,\ s_2=6,\ t_2=1$.
Since $\frac{\gamma^{\top}B\gamma}{2}=\begin{pmatrix}
    0 & \frac{1}{2} \\ -\frac{1}{2} & 0
\end{pmatrix}$, the matrix $\epsilon$ is $\begin{pmatrix}
    0 & 1 \\ -3 & 0
\end{pmatrix}$. 
Then the twist equation becomes 
\begin{equation}\label{eq3.73}
    \begin{split}
        \frac{r_1}{2}+\left(\frac{\beta_1}{2}+\frac{\beta_2}{2}\right)&=0\mod 1,\\
        \frac{r_2}{6}+\left(\frac{\beta_2}{6}-\frac{\beta_1}{2}\right)&=0\mod 1,
    \end{split}
\end{equation}
for $(r_1,r_2,\beta_1,\beta_2)\in\mathbb{Z}_4 \times\mathbb{Z}_{12}\times\mathbb{Z}_2 \times\mathbb{Z}_6$,
and this equations yields the solutions as in the table \ref{table2}.
\begin{table}[h]
    \caption{The solutions of eq.\eqref{eq3.73}}
    \centering
    \begin{tabular}{cc}
        \hline
        $\vec{\beta}$&$\vec{r}$\\
        \hline\hline
        $\begin{pmatrix}
            0 \\ 0
        \end{pmatrix},\begin{pmatrix}
            1\\3
        \end{pmatrix}$
        &
        $\begin{pmatrix}
            0 \\ 0
        \end{pmatrix},\begin{pmatrix}
            0\\6
        \end{pmatrix},\begin{pmatrix}
            2\\0
        \end{pmatrix},\begin{pmatrix}
            2\\6
        \end{pmatrix}$ 
        \\
        $\begin{pmatrix}
            0\\1
        \end{pmatrix},\begin{pmatrix}
            1\\4
        \end{pmatrix}$
        &
        $\begin{pmatrix}
            1\\5
        \end{pmatrix},\begin{pmatrix}
            1\\11
        \end{pmatrix},\begin{pmatrix}
            3\\5
        \end{pmatrix},\begin{pmatrix}
            3\\11
        \end{pmatrix}$\\
        $\begin{pmatrix}
            0\\2
        \end{pmatrix},\begin{pmatrix}
            1\\5
        \end{pmatrix}$&
        $\begin{pmatrix}
            0\\4
        \end{pmatrix},\begin{pmatrix}
            0\\10
        \end{pmatrix},\begin{pmatrix}
            2\\4
        \end{pmatrix}\begin{pmatrix}
            2\\10
        \end{pmatrix}
        $\\
        $\begin{pmatrix}
            0\\3
        \end{pmatrix},\begin{pmatrix}
            1\\0
        \end{pmatrix}$&
        $\begin{pmatrix}
            1\\3
        \end{pmatrix},\begin{pmatrix}
            1\\9
        \end{pmatrix},\begin{pmatrix}
            3\\3
        \end{pmatrix},\begin{pmatrix}
            3\\9
        \end{pmatrix}$\\
        $\begin{pmatrix}
            0\\4
        \end{pmatrix},\begin{pmatrix}
            1\\1
        \end{pmatrix}$&
        $\begin{pmatrix}
            0\\2
        \end{pmatrix},\begin{pmatrix}
            0\\8
        \end{pmatrix},\begin{pmatrix}
            2\\2
        \end{pmatrix},\begin{pmatrix}
            2\\8
        \end{pmatrix}$\\
        $\begin{pmatrix}
            0\\5
        \end{pmatrix},\begin{pmatrix}
            1\\2
        \end{pmatrix}$&
        $\begin{pmatrix}
            1\\1
        \end{pmatrix},\begin{pmatrix}
            1\\7
        \end{pmatrix},\begin{pmatrix}
            3\\1
        \end{pmatrix},\begin{pmatrix}
            3\\7
        \end{pmatrix}$\\
        \hline
    \end{tabular}
    \label{table2}
\end{table}

Consequently, $\mathcal{C}(\mathcal{A}')$ is found to consist of the following codewords
\begin{equation*}
    \begin{split}
        &(0,0,0,0),(0,1,1,0),(0,2,0,0),(0,3,1,0)\\
        &(0,4,0,0),(0,5,1,0),(1,0,0,3),(1,1,1,3)\\
        &(1,2,0,3),(1,3,1,3),(1,4,0,3),(1,5,1,3),
    \end{split}
\end{equation*}
which constitute the additive subgroup of 
$\mathbb{Z}_2 \times\mathbb{Z}_6 \times\mathbb{Z}_2 \times\mathbb{Z}_6$. As this 
subgroup is generated by $(1,0,0,3)$ and $(0,1,1,0)$, the generator matrix $G$ of $\mathcal{C}(\mathcal{A}')$
can be expressed as 
\begin{equation}
    G=\begin{pmatrix}
        1&0&0&3\\0&1&1&0
    \end{pmatrix},
\end{equation}
and this is a B-form code.

To close this section, let us discuss the metric of non-homogeneous codes.
We will first recall the $\mathbb{F}_p$ code $\mathcal{C}$. 
The self-duality of $\mathcal{C}$ reflects the self-duality of the lattice constructed 
from $\mathcal{C}$, because 
\begin{equation}
    \begin{split}
     \frac{\alpha +pk}{\sqrt{p}}\cdot\frac{\beta+pl}{\sqrt{p}}
    +\frac{\alpha' +pk'}{\sqrt{p}}\cdot\frac{\beta' +pl'}{\sqrt{p}}&=0\mod 1\\
    \Leftrightarrow
    \alpha\cdot\beta' + \alpha'\cdot\beta &=0 \mod p
    \end{split}
\end{equation} 
where $(\alpha_1,\beta_1),(\alpha_2,\beta_2)\in\mathcal{C}$ and $k,k',l,l'\in\mathbb{Z}^n$.
Then, consider the non-homogeneous case where $\mathcal{C}$ is an additive 
subgroup of $(\mathbb{Z}_{d_1}\times\cdots\times\mathbb{Z}_{d_n})^2$.
Let $(\alpha,\beta)$ be elements of $\mathcal{C}$.
The lattice vector corresponding to $(\alpha,\beta)$ in Construction A can be written as
$(\frac{\alpha_1+d_1 k_1}{\sqrt{d_1}},\dots,\frac{\alpha_n+d_n k_n}{\sqrt{d_n}},
\frac{\beta_1+d_1 k_1'}{\sqrt{d_1}},\dots,\frac{\beta_n+d_n k_n'}{\sqrt{d_n}})$ 
with integers $k_i,k_i'$.
Therefore the self-duality of the Narain lattice implies that for another vector 
$(\alpha',\beta')$,
\begin{equation}
    \begin{split}
        \frac{\alpha_1\beta_1'+\alpha'_1\beta_1 +d_1 K_1}{d_1}+\cdots+
        \frac{\alpha_n\beta'_n+\alpha'_n\beta_n+d_n K_n}{d_n}&=0\mod 1   \\
        \Leftrightarrow\frac{L}{d_1}(\alpha_1\beta_1'+\alpha'_1\beta_1)+\cdots
        +\frac{L}{d_n}(\alpha_n\beta'_n+\alpha'_n\beta_n)&=0\mod L
    \end{split}
\end{equation}
for some integers $K_i$ and $L=\lcm(d_1,\dots,d_n)$.
From this discussion, one finds that defining the metric matrix $g$ as 
\begin{equation}
    g=\begin{pmatrix}
         & & & \frac{L}{d_1} & & \\
         & \text{\huge{0}} & & & \ddots & \\
         & & & & & \frac{L}{d_n} \\
        \frac{L}{d_1} & & & & & \\
         & \ddots & & & \text{\huge{0}} & \\
         & & \frac{L}{d_n} & & & 
    \end{pmatrix},
\end{equation}
the self-duality of $\mathcal{C}$ can be defined as 
$GgG^{\top}=0\mod L$,
in consistent with the homogeneous case.
Moreover, defining the distance $d'(\mathcal{C})$ of $\mathcal{C}$ as
\begin{equation}
    d'(\mathcal{C})\coloneqq \min_{0\neq(\alpha,\beta)\in\mathcal{C}}\left\{
    \sum_{i=1}^n \frac{1}{d_i}(\min\{\alpha_i^2 ,(d_i -\alpha_i)^2\}+
    \min\{\beta_i^2 ,(d_i -\beta_i)^2\})
    \right\},
\end{equation}
the spectral gap $\Delta$ of the code CFT and $d'(\mathcal{C})$ is related:
\begin{equation}
    \Delta=\frac{1}{2}\min\{d'(\mathcal{C}),\min_{i}\{d_i\}\}.
\end{equation}

\section{Discussion and Further Direction}\label{sec4}
So far, we have searched for physical or code-theoretical translations 
of simple current orbifold discrete torsions. In other words, a discrete 
torsion roughly reflects the B-field information using the coordinates of 
the Narain moduli, and if we consider it to be the B-form part of a B-form code, 
the partition function can be constructed from the weight enumerator of the code. 
Our next task is how to recover the spectral gap of the Narain RCFT from this discrete 
torsion. We expect that this can be accomplished with an appropriate extension of our 
previous paper. That is, we consider the discrete torsion as a generalization of the 
Boolean function. We are also interested in graph-theoretic constructions 
of such theories.

The case we discussed was the case where $\gamma^{\top}\gamma$ is a diagonal matrix. 
However, in general, this is not a diagonal matrix, but a real symmetric matrix. 
A real symmetric matrix can be diagonalized by an orthogonal matrix $P$. Nevertheless, 
if $P$ does not take integer values, this discussion cannot be applied as it is. 
(The case of integer-valued P, which corresponds to a permutation matrix, is not discussed.) 
The reason is that when $N$ and $M$ are transformed by $P$, 
the resulting elements are generally not integers. It would be an intriguing problem 
to extend the discussion to encompass this aspect as well. If this is resolved, the 
findings of this paper may be applicable to code CFT, as discussed by Angelions \cite{angelinos2022}. 
However, in that case, the code would be an isodual code instead of a self-dual code, 
necessitating further discussion.
Once this is done, the original goal of "identifying the combinatorial 
structure that should be present in the partition function of all rational Narain CFTs" 
is nearly achieved. This allows us to identify their holographic behavior by computing the sum 
of all rational Narain theories in a combinatorial way. Furthermore, 
rational CFT has been mathematically related to 3D TQFT, which is expected to 
lead to a mathematical clarification of the correspondence between ensemble average 
and quantum gravity.

We discussed a code CFT that can reproduce not only conventional codes but also the 
structure of codes in which each bit has a distinct dimension. Such a code can be 
formulated as a code over a ring. It is an interesting problem to perform averaging 
over codes with self duality, 
and we expect to obtain 
results that encompass the averaging that has been considered so far. Another work on 
this code CFT is \cite{kawabata2023,henriksson2022,dymarsky2021a}, 
and it would be interesting to apply it to these applications. There are 
other interesting works \cite{ashwinkumar2021,henriksson2023} where the notion 
of Narain CFTs is extended.


In addition, the discrete torsion that has appeared in this article is mathematically 
formulated as Kreuzer Schellekens Bihomomorphism (KSB). In RCFT terminology, 
this is a map that gives a gluing together of a chiral part and an 
antichiral part, which is classified using the cohomology of the group. 
In a previous paper, the author introduced the Boolean function f to write the 
spectral gap $\Delta$ in terms of f and discussed the consideration of f where $\Delta$ is large. 
Since this KSB can be regarded as a generalization of the Boolean function, 
we intend to set up a future work to see if in the general RCFT $\Delta$ can be written using KSB.
\section*{Acknowledgements}
We thank Toshiya Kawai, 
Anatoly Dymarsky, and Masahito Yamazaki for discussions and comments on our paper. This work is 
supported by JST the establishment of university fellowship towards the 
creation of science technology innovation grant No.JPMJF2123 and JSPS KAKENHI grant No.23KJ1183.

\newpage
\bibliographystyle{jalpha}
\bibliography{Furuta202304}
\end{document}